\documentclass[10pt,draftcls,onecolumn]{IEEEtran}
\usepackage{amsmath}
\usepackage{amssymb}
\usepackage{cite}
\usepackage{color}
\usepackage{epsfig}
\usepackage{epsf}
\usepackage{rotating}
\usepackage{mathrsfs}
\usepackage{graphics}
\usepackage{theorem}

\newtheorem{Theorem}{Theorem}

\newtheorem{corollary}{corollary}
\newtheorem{Proposition}{Proposition}
\newtheorem{Remark}{Remark}

\begin{document}

%
%
%
%
%
%


\title{Cooperative Strategies for the Half-Duplex Gaussian Parallel Relay Channel:\\
Simultaneous Relaying versus Successive Relaying \footnote{Financial supports provided by Nortel, and the corresponding matching
 funds by the Federal government: Natural Sciences and Engineering Research Council of Canada (NSERC)
 and Province of Ontario: Ontario Centres of Excellence (OCE) are gratefully acknowledged.}}

\author{\normalsize
Seyed Saeed Changiz Rezaei, Shahab Oveis Gharan, and Amir K. Khandani \\
\small Coding \& Signal Transmission Laboratory\\[-5pt]
\small Department of Electrical \& Computer Engineering\\[-5pt]
\small University of Waterloo \\[-5pt]
\small Waterloo, ON, N2L\ 3G1 \\[-5pt]
\small {sschangi, shahab, khandani}@cst.uwaterloo.ca\\}

\date{}
\maketitle \thispagestyle{empty}

\begin{abstract}
This study investigates the problem of communication for a network composed of two half-duplex parallel relays with additive white Gaussian noise. Two protocols, i.e., \emph{Simultaneous} and \emph{Successive} relaying, associated with two possible relay orderings are proposed. The simultaneous relaying protocol is based on \emph{Dynamic Decode and Forward (DDF)} scheme. For the successive relaying protocol: (i) a \emph{Non-Cooperative} scheme based on the \emph{Dirty Paper Coding (DPC)}, and (ii)  a \emph{Cooperative} scheme based on the \emph{Block Markov Encoding (BME)} are considered. Furthermore, the composite scheme of employing BME at one relay and DPC at another always achieves a better rate when compared to the \emph{Cooperative} scheme. A \emph{``Simultaneous-Successive Relaying based on Dirty paper coding scheme" (SSRD)} is also proposed. The optimum ordering of the relays and hence the capacity of the half-duplex Gaussian parallel relay channel in the low and high signal-to-noise ratio (SNR) scenarios is derived. In the low SNR scenario, it is revealed that under certain conditions for the channel coefficients, the ratio of the achievable rate of the simultaneous relaying based on DDF to the cut-set bound tends to be 1. On the other hand, as SNR goes to infinity, it is proved that successive relaying, based on the DPC, asymptotically achieves the capacity of the network.
\end{abstract}
%
%
\section{Introduction}

\subsection{Motivation}

 The continuous growth in wireless
communication has motivated information theoretists to extend shannon's information theoretic arguments for a single user channel to the scenarios that involve communication among multiple users.

 In this regard, cooperative wireless communication has been the focus of attention during recent years. Due to rapid decrease of the transmitted signal power with distance, the idea of multi-hopped communication has been proposed. In multi-hopped communication, some intermediate nodes as relays are exploited to facilitate data transmission from the source to the destination. Using this technique leads to saving battery power as well as increasing the physical coverage area. Moreover, relays by emulating distributed transmit antenna, can form spatial diversity and combat the multi-path fading effect of the wireless media.

  Motivated by practical constraints, half-duplex relays which cannot
transmit and receive at the same time and in the same frequency band
are of great importance. Here, our goal is to study and analyze the
performance limits of a half-duplex parallel relay channel.

\subsection{History}
Relay channel is a three terminal network which was
introduced for the first time by Van der Meulen in 1971
\cite{VanderMeulen}. The most important capacity results of the relay
channel were reported by Cover and El Gamal \cite{Cover}. Two relaying strategies are proposed in \cite{Cover}. In one strategy, the relay decodes the transmitted message and forwards the re-encoded version to the destination, while in another one the relay does not
decode the message, but sends the quantized received values to the
destination.

Moreover, several works on multi-relay channels exist in the
literature (See \cite{schein1,schein,Xie2,Gastpar3,Gastpar4,Gastpar5,Shamai1,Yates1,PeyWei,Gastpar2,Wittneben1,Wittneben2,Sumeet,Yong,Belfiore,Azarian,Mitran,Poor}). Schein in~\cite{schein1,schein} establishes upper and lower
bounds on the capacity of a full-duplex parallel relay channel in which
the channel consists of a source, two relays and a destination, where there is no direct link between the source and the destination, and also between the two relays. Generally, the
best rate reported for the full-duplex Gaussian parallel relay
channel is based on the Decode-Forward (DF) or Amplify-Forward (AF) schemes,
with time sharing \cite{schein1,schein}. Xie and Kumar generalize the block
Markov encoding scheme in \cite{Cover} for a network of multiple relays \cite{Xie2}.
Gastpar, Kramer, and
Gupta extend compress and forward scheme to a multiple relay
channel by introducing the concept of antenna polling in
\cite{Gastpar3,Gastpar4,Gastpar5}. In \cite{Shamai1}, Amichai,
Shamai, Steinberg and Kramer consider a parallel relay setup, in which a nomadic
source sends its information to a remote destination via some relays with
lossless links to the destination. They investigate the case that these relays
do not have any decoding capability, so signals received at the relays must be compressed. The authors also fully characterize the capacity of this case for the
Gaussian channel. In \cite{Yates1}, Maric and Yates investigate DF
and AF schemes in a parallel-relay network. Motivated by
applications in sensor networks, they assume large bandwidth
resources allowing orthogonal transmissions at different nodes. They
characterize optimum resource allocation for AF and DF and show
that the wide-band regime minimizes the energy cost per information
bit in DF, while AF should work in the band-limited regime to achieve the
best rate. Razaghi and
Yu in \cite{PeyWei} propose a parity-forwarding scheme for
full-duplex multiple relay. They show that parity-forwarding can achieve the capacity
in a new form of degraded relay networks.

Radios that can receive and transmit simultaneously in the same frequency band require complex and expensive components \cite{HostMadsen}. Hence,
Khojastepour and Aazhang in \cite{KhojastepourAazhang1},
\cite{KhojastepourAazhang2} call the half-duplex relay as ``\emph{Cheap
Relay}".

Recently, half-duplex relaying has drawn a great deal of attention (See \cite{KhojastepourAazhang1,KhojastepourAazhang2,szahedi,SinaZahedi,ElGamal,HostMadsen,LiangVeeravalli}, \cite{Gastpar2}, \cite{Wittneben1,Wittneben2,Sumeet,Yong,Belfiore,Azarian,Mitran,Poor}).
Zahedi and El Gamal consider two different cases of frequency
division Gaussian relay channel, deriving lower and upper bounds
on the capacity \cite{szahedi}. They also derive single letter characterization of
the capacity of frequency division additive white Gaussian noise (AWGN) relay channel with simple
linear relaying scheme \cite{SinaZahedi},\cite{ElGamal}. The problem
of time division relaying is also considered by Host-Madsen
and Zhang \cite{HostMadsen}. By considering fading scenarios, and assuming channel
state information (CSI), they study upper and lower bounds on the outage
capacity and the Ergodic capacity. In \cite{LiangVeeravalli}, Liang
and Veeralli present a Gaussian orthogonal relay model, in which
the relay-to-destination channel is orthogonal to the source-to-relay and source-to-destination channel.
They show that when the source-to-relay channel is better
than the source-to-destination channel and the signal-to-noise ratio (SNR)
of the relay-to-destination is less than a given threshold,
optimizing resource allocation causes the lower and the upper bounds to
coincide with each other.

\subsection{Contributions and Relation to Previous Works}

In this paper, we study transmission strategies for a network with a
source, a destination, and two half-duplex relays with additive white Gaussian noise which cooperate with each
other to facilitate data transmission from the source to the
destination. Furthermore, it is assumed that no direct link exists between the source and the destination.

Half-duplex relaying, in multiple relay networks, is studied in \cite{Gastpar2,Wittneben1,Wittneben2,Sumeet,Yong,Belfiore,Azarian,Mitran,Poor}. Gastpar in \cite{Gastpar2} shows that in a Gaussian parallel
relay channel with infinite number of relays, the optimum coding
scheme is AF. Rankov and Wittneben in \cite{Wittneben1,Wittneben2} further study the problem of half-duplex
relaying in a two-hop communication scenario. In their study, they also consider a parallel relay setup with two relays where there is no direct link between the source and the destination, while there exists a link between the relays.
Their relaying protocols are based on either AF or DF, in which the relays
successively forward their messages from the source to the destination. We
call this protocol ``\emph{Successive Relaying}" in the sequel.
Xue and Sandhu in~\cite{Sumeet} further study different
half-duplex relaying protocols for the Gaussian parallel relay channel. Since they assume that
there is no link between the relays, they refer to their parallel
channel as a \emph{Diamond Relay Channel}.

In this work, our primary objective
is to find the best ordering of the relays in the intended set-up. We consider two relaying
protocols, i.e., simultaneous relaying versus successive relaying,
associated with two possible relay orderings. For simultaneous relaying, each relay exploits
``Dynamic DF (DDF)". It should be noted that the DDF scheme considered here is slightly different from the DDF introduced in \cite{Azarian} and \cite{Mitran}. In those works, the DDF scheme is applied to the set-up of the multiple relay network in which the nodes only have the CSI of their receiving channel. In the DDF scheme described in \cite{Azarian}, the source is broadcasting the message to all the network nodes during whole period of transmission and each relay, listens to the transmitted signal of the source and other relays until it can decode the transmitted message. Consequently, it transmits its signal coherently with the source and other active relays in the remaining time. However, in our set-up, all the nodes are assumed to have all the channel coefficients. Therefore, in a fixed pre-assigned portion of the time, the relays receive the signal transmitted from the source, and in the remaining time slot they transmit the re-encoded version of the decoded message together. In other words, the relays operate in a synchronous manner.

For successive relaying, we study a \emph{Non-Cooperative} scheme based on ``Dirty Paper Coding (DPC)" and also a \emph{Cooperative} scheme based on ``Block Markov Encoding (BME)". It is worth noting that the authors in ~\cite{Poor} also propose successive relaying protocol for the set up with two parallel relays and direct links between the relays and between the source and the destination. They propose a simple repetition coding at the relays, and show that their scheme can recover the loss in the multiplexing gain, while achieving diversity gain of 2.

We derive the optimum relay ordering in low and high SNR scenarios. In low SNR scenarios and under certain channel conditions, we show that the ratio of the achievable rate of DDF for simultaneous relaying to the cut-set bound tends to one. On the other hand, in high SNR scenarios, we prove that the proposed DPC for successive relaying asymptotically achieves the capacity.

After this work was completed, we became aware of~\cite{Yong} which has independently
proposed an achievable rate based on the combination of superposition coding,
  BME and DPC. In their scheme, the intended message $``w"$ is split into a message which is transmitted to the destination by exploiting cooperation between the relays $``w_r"$ and a message which is transmitted to the destination without using any cooperation between the relays $``w_d"$. Hence, the signal associated with $``w_d"$, transmitted by one relay, can be considered as interference on the other relay. $``w_r"$ is transmitted by using BME and $``w_d"$ is transmitted by employing DPC. Therefore, in their general scheme, the associated signals with these two messages are superimposed and transmitted. As the channel between the two relays become strong, their proposed scheme is converted to BME. On the other hand, as the channel becomes weak, their proposed scheme becomes DPC.

  Unlike~\cite{Yong}, in which the authors only consider successive relaying and propose a combined BME and DPC, as the main result of this paper, simultaneous and successive relaying protocols are combined and a ``Simultaneous-Successive Relaying based on Dirty paper coding" (SSRD) scheme with a new achievable rate is proposed. It is shown that in the low SNR scenario and under certain channel conditions, SSRD scheme is converted to simultaneous relaying based on DDF, while in the high SNR scenarios, when the ratio of the relay powers to the source power remain constant, it becomes successive relaying based on DPC (to achieve the capacity).

Besides this main result, some other results obtained in this paper are as follows:
  \begin{itemize}
  \item Two different types of decoding, i.e., \emph{successive} and \emph{backward} decoding, at the destination for the BME scheme are proposed. We prove that the achievable rate of BME with backward decoding is greater than that of BME with successive decoding, i.e., $C_{BME_{back}}^{low}\geq C_{BME_{succ}}^{low}$.

   \item It is proved that BME with backward decoding leads to a simple strategy in which at most, one of the relays is required to cooperate with the other relay in sending the bin index of the other relay's message. Accordingly, in the Gaussian case, the combination of BME at one relay and DPC at the other relay always achieves a better rate than the simple BME.

  \item In the degraded case, where the destination receives a degraded version of the received signals at the relays, BME with backward decoding achieves the successive cut-set bound.
  \end{itemize}

The rest of the paper is organized as follows: In section II, the
system model is introduced. In section III, the achievable rates and coding schemes for
a half-duplex relay network are derived. Optimality results are discussed in section IV. Simulation
results are presented in section V. Finally, section VI
concludes the paper.

\subsection{Notation}

Throughout the paper, the superscript $^H$ stands for matrix operation of conjugate
transposition. Lowercase bold letters and regular letters
represent vectors and scalars, respectively. For any two functions
$f(n)$ and $g(n)$, $f(n)=O(g(n))$ is equivalent to $\lim_{n
\rightarrow \infty} \left| \frac{f(n)}{g(n)} \right| < \infty$, and
$f(n)=\Theta(g(n))$ is equivalent to $\lim_{n \rightarrow \infty}
\frac{f(n)}{g(n)} =c$, where $0<c<\infty$. And $C(x)\triangleq\frac{1}{2}\log_2(1+x)$. Furthermore, for the sake of brevity, $A_{\epsilon}^{(n)}$ denotes the set of weakly jointly typical sequences for any intended set of random variables.

\section{System Model}
We consider a Gaussian network which consists of a source, two half-duplex relays, and a destination, and there is no direct link between the source and the destination. Here we define four time slots according to the transmitting and receiving mode of each relay (See Fig.~\ref{fig:SysModelm.eps}), where $t_b$ denotes the duration of time slot $b$ ($\sum_{b=1}^4 t_b = 1$). Nodes 0, 1, 2, and 3 represent the source, relay 1, relay 2, and the destination, respectively. Moreover, the transmitting and receiving signals at node $a$ during time slot $b$ are represented by $\textbf{x}_a^{(b)}$ and $\textbf{y}_a^{(b)}$, respectively. Hence, at each node $c\in\{1,2,3\}$, we have
\begin{eqnarray}
\textbf{y}_{c}^{(b)} = \sum_{a\in\{0,1,2\}} h_{ac}\textbf{x}_{a}^{(b)} + \textbf{z}_{c}^{(b)}.
\end{eqnarray}
where $h_{ac}$$^{,}$s denote channel coefficients from node $a$ to node $c$, and $\textbf{z}_{c}^{(b)}$ is the AWGN term with zero mean and variance of $``1"$ per
dimension.
\begin{figure}[hbtp]
\centering
\includegraphics[angle=0,scale=.75]{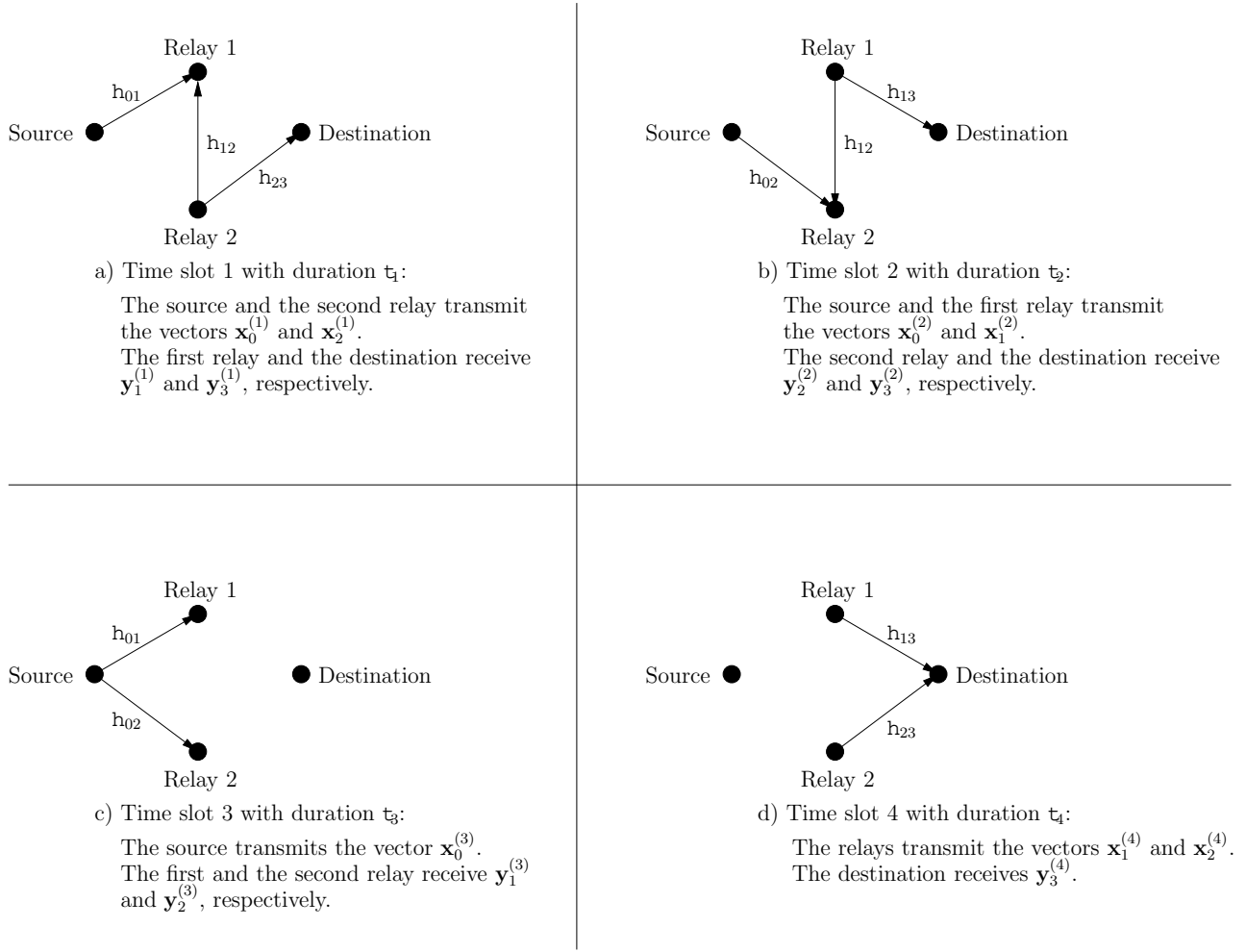}\\
\caption{\small{System Model.}}\label{fig:Sysmodelm.eps}
\end{figure}

Noting the transmission strategies in Fig.~\ref{fig:SysModelm.eps}, we have
\begin{IEEEeqnarray}{rl}
&\textbf{y}_{1}^{(1)} = h_{01}\textbf{x}_{0}^{(1)} + h_{21}\textbf{x}_{2}^{(1)} + \textbf{z}_{1}^{(1)},\label{eq:Sys1_1}\\
&\textbf{y}_{3}^{(1)} = h_{23}\textbf{x}_{2}^{(1)} + \textbf{z}_{3}^{(1)},\label{eq:Sys1_2}\\
&\textbf{y}_{2}^{(2)} = h_{02}\textbf{x}_{0}^{(2)} + h_{12}\textbf{x}_{1}^{(2)} + \textbf{z}_{2}^{(2)},\label{eq:Sys2_1}\\
&\textbf{y}_{3}^{(2)} = h_{13}\textbf{x}_{1}^{(2)} + \textbf{z}_{3}^{(2)},\label{eq:Sys2_2}\\
&\textbf{y}_{k}^{(3)} = h_{0k}\textbf{x}_{0}^{(3)} + \textbf{z}_{k}^{(3)}, k\in\{1,2\},\label{eq:Sys3}\\
&\textbf{y}_{3}^{(4)} = \sum_{k = 1}^{2}h_{k3}\textbf{x}_{k}^{(4)} + \textbf{z}_{3}^{(4)}.\label{eq:Sys4}
\end{IEEEeqnarray}
 Throughout the paper, we assume that $h_{01}\geq h_{02}$ unless specified otherwise, and from reciprocity assumption, we have $h_{12}=h_{21}$.
Furthermore, the power constraints $P_0$, $P_1$, and $P_2$ should be satisfied for the source, the first relay, and the second relay, respectively. Hence, denoting the power consumption of node $a$ at time slot $b$ by $P_a^{(b)} = E\left[\textbf{x}_a^{(b)H}\textbf{x}_a^{(b)}\right]$, we have
\begin{IEEEeqnarray}{rl}
&P_{0}^{(1)}+P_{0}^{(2)}+P_{0}^{(3)} = P_0,\\
&P_{1}^{(2)}+P_{1}^{(4)} = P_1,\nonumber\\
&P_{2}^{(1)}+P_{2}^{(4)} = P_2.\nonumber
\end{IEEEeqnarray}

\section{Achievable Rates and Coding Schemes}\label{section:proposed_method}

In this section, we propose two cooperative protocols, i.e. \emph{Successive} and \emph{Simultaneous} relaying protocols, for a half-duplex Gaussian parallel relay channel.
\subsection{Successive Relaying Protocol}
In \emph{Successive} relaying protocol, the relays are not allowed to receive and transmit simultaneously, i.e. $t_3 = t_4 = 0$, and the relations between the transmitted and the received signals at the relays and at the destination follow from (\ref{eq:Sys1_1})-(\ref{eq:Sys2_2}). For the successive relaying protocol, we propose a \emph{Non-Cooperative} and a \emph{Cooperative Coding} scheme in the sequel.
In the proposed schemes, the time is divided into odd and even time slots with the duration $t_1$ and $t_2$, respectively. Accordingly, at each odd and even time slots, the source transmits a new message to one of the relays, and the destination receives a new message from the other relay, successively (See Fig. \ref{fig:p1_4.eps}).
\begin{figure}[hbtp]
\centering
\includegraphics[angle=0,scale=.75]{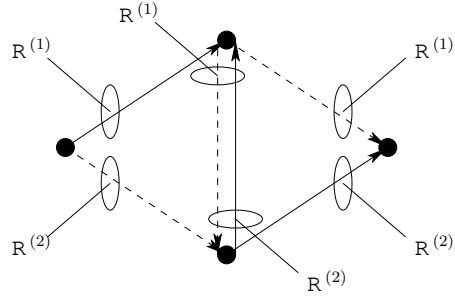}\\
\caption{\small{Information flow transfer for successive relaying protocol for two relays.}}\label{fig:p1_4.eps}
\end{figure}
\subsubsection{Non-Cooperative Coding}
In the Non-Cooperative Coding scheme, each relay considers the other's signal as
interference. Since the source knows each relay's message, it
can apply the Gelfand-Pinsker's coding scheme to transmit its
message to the other relay. The following Theorem gives the achievable rate of this scheme.

\begin{figure}[hbtp]
\centering
\includegraphics[angle=0,scale=.75]{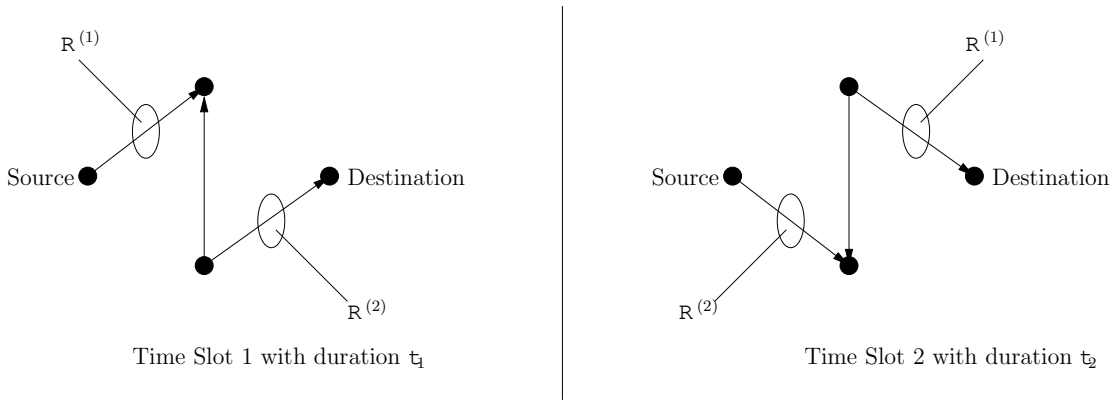}\\
\caption{\small{Successive relaying protocol based on Non-Cooperative Coding.}}\label{fig:p1_8m.eps}
\end{figure}

\begin{Theorem}\label{thm:Th2}
For the half-duplex parallel relay channel, assuming successive
relaying, the following rate $C_{DPC}^{low}$ is achievable:
\begin{IEEEeqnarray}{rl}
C_{DPC}^{low} = &\max_{0\leq t_{1},t_{2},t_{1}+t_{2}=1} R^{(1)} + R^{(2)},\label{eq:DPC1}\\
&\text{subject to:}\nonumber\\
&~~~~~R^{(1)}\leq \min\left(t_{1}(I(U_{0}^{(1)};Y_{1}^{(1)})-I(U_{0}^{(1)};X_{2}^{(1)})),
t_{2}I(X_{1}^{(2)};Y_{3}^{(2)})\right),\label{eq:DPC2}\\
&~~~~~R^{(2)}\leq \min\left(t_{2}(I(U_{0}^{(2)};Y_{2}^{(2)})-I(U_{0}^{(2)};X_{1}^{(2)})),t_{1}I(X_{2}^{(1)};Y_{3}^{(1)})\right).\label{eq:DPC3}
\end{IEEEeqnarray}
with probabilities:
\begin{IEEEeqnarray}{rl}
&p(x_{2}^{(1)},u_{0}^{(1)},x_{0}^{(1)})=p(x_{2}^{(1)})p(u_{0}^{(1)}|x_{2}^{(1)})p(x_{0}^{(1)}|u_{0}^{(1)},x_{2}^{(1)}),\nonumber\\
&p(x_{1}^{(2)},u_{0}^{(2)},x_{0}^{(2)})=p(x_{1}^{(2)})
p(u_{0}^{(2)}|x_{1}^{(2)})p(x_{0}^{(2)}|u_{0}^{(2)},x_{1}^{(2)}).\nonumber
\end{IEEEeqnarray}
\end{Theorem}
\begin{proof}
See Appendix A.
\end{proof}

From Theorem \ref{thm:Th2}, the achievable rate of the proposed scheme for the Gaussian case can be obtained as follows.

\begin{corollary}\label{cor:cor2}
For the half-duplex Gaussian parallel relay channel, assuming
successive relaying protocol with power constraint at the source and
at each relay, DPC achieves the following rate:
\begin{IEEEeqnarray}{rl}
C_{DPC}^{low}&=\max\left(R^{(1)} + R^{(2)}\right),\label{eq:R_DPC_1}\\
&\text{subject to:}\nonumber\\
&~~~~~R^{(1)}\leq\min\left(t_{1}C\left(\frac{h_{01}^2P_{0}^{(1)}}{t_{1}}\right),
t_{2}C\left(\frac{h_{13}^2P_1}{t_{2}}\right)\right),\nonumber\\
&~~~~~R^{(2)}\leq\min\left(t_{2}C\left(\frac{h_{02}^2P_{0}^{(2)}}{t_{2}}\right),
t_{1}C\left(\frac{h_{23}^2P_2}{t_{1}}\right)\right),\nonumber\\
&~~~~~P_{0}^{(1)} + P_{0}^{(2)} = P_0,\nonumber\\
&~~~~~t_{1} + t_{2} = 1,\nonumber\\
&~~~~~0 \leq t_1, t_2, P_0^{(1)}, P_0^{(2)}.\nonumber
\end{IEEEeqnarray}
\end{corollary}
\begin{proof}
From Costa's Dirty Paper Coding~\cite{Costa}, by having
\begin{eqnarray}\label{eq:DIRT}
U_{0}^{(1)} = X_{0}^{(1)} + \frac{h_{01}h_{12}P_{0}^{(1)}}{h_{01}^2P_{0}^{(1)} + t_{1}}X_{2}^{(1)},\\
U_{0}^{(2)} = X_{0}^{(2)} + \frac{h_{02}h_{12}P_{0}^{(2)}}{h_{02}^2P_{0}^{(2)} + t_{2}}X_{1}^{(2)}.
\end{eqnarray}
where
$X_{0}^{(1)}\sim\mathcal{N}(0,P_{0}^{(1)}),X_{0}^{(2)}\sim\mathcal{N}(0,P_{0}^{(2)}),X_{2}^{(1)}\sim\mathcal{N}(0,P_2)$,
and $X_{1}^{(2)}\sim\mathcal{N}(0,P_1)$, and applying
them to Theorem~\ref{thm:Th2}, we obtain corollary~\ref{cor:cor2}.
\end{proof}

\subsubsection{Cooperative Coding}

\begin{figure}[hbtp]
\centering
\includegraphics[angle=0,scale=.65]{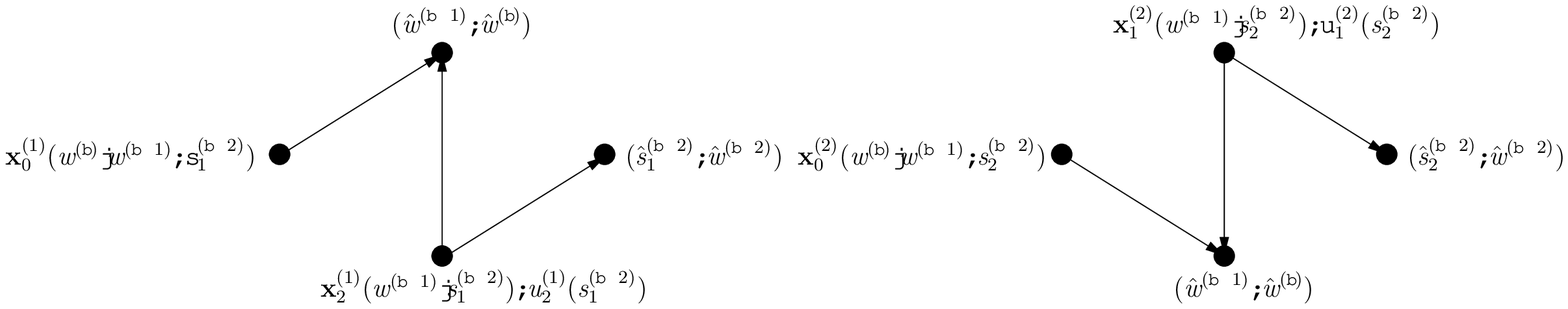}\\
\caption{\small{Successive relaying protocol based on Cooperative Coding.}}\label{fig:p1_3.eps}
\end{figure}
\begin{figure}[hbtp]
\centering
\includegraphics[angle=0,scale=.75]{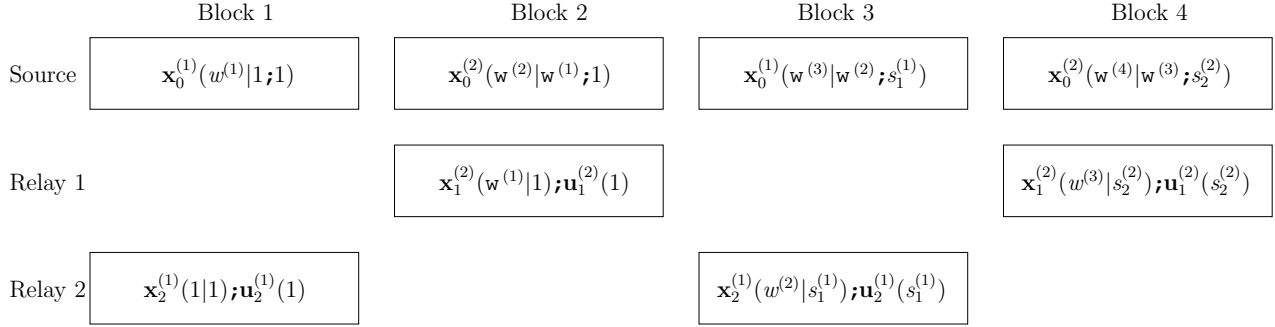}\\
\caption{\small{Decode-and-forward for successive relaying protocol.}}\label{fig:p1_2.eps}
\end{figure}
In this type of coding scheme, we assume that, at each time slot, the receiving relay decodes not only the new transmitted message from the source, but also the previous message transmitted from the transmitting relay (See Figs.~\ref{fig:p1_4.eps} and \ref{fig:p1_3.eps}).
Our proposed coding scheme is based on binning, superposition
coding, and Block Markov Encoding. The source sends $B$ messages $w^{(1)},w^{(2)},\cdots,w^{(B)}$ in $B+2$ time slots.

Generally, this scheme can be described as follows (See Figs.~\ref{fig:p1_3.eps} and~\ref{fig:p1_2.eps}). In time slot $b$, the relay $(b+1)$ mod 2$+1$ decodes the transmitted messages $w^{(b)}$ and $w^{(b-1)}$
from the source and the other relay, respectively. In time slot $b+1$, it broadcasts $w^{(b)}$ and the bin index of $w^{(b-1)}$, $s_{(b+2)\text{ mod }2+1}^{(b-1)}$, to the destination using the binning function defined next.

\emph{Definition (The Binning Function):} The
binning function $f_{Bin}^{({(b+1)\text{ mod }2+1})}(w^{(b-2)}) :W=
\{1,2,\cdots,2^{nR^{({(b+1)\text{ mod }2+1})}}\}
\\ \longrightarrow\{1,2,\ldots,2^{\emph{nr}_{Bin}^{({(b+1)\text{ mod }2+1})}}\}$ is defined
by $f_{Bin}^{({(b+1)\text{ mod }2+1})}(w^{(b-2)}) = s_{(b+1)\text{ mod }2+1}^{(b-2)}$, where $f_{Bin}^{({(b+1)\text{ mod }2+1})}(.)$ assigns a randomly uniform distributed integer between 1 and $2^{nr_{Bin}^{({(b+1)\text{ mod }2+1})}}$ independently to each member of $W$.

As indicated in Fig.~\ref{fig:p1_2.eps}, in the first time slot,
the source transmits the codeword $\textbf{x}_{0}^{(1)}(w^{(1)}|1,1)$
to the first relay, while the second relay transmits a doubly
indexed codeword $\textbf{x}_{2}^{(1)}(1|1)$ and the codeword
$\textbf{u}_{2}^{(1)}(1)$ to the first
relay and to the destination. In the second time slot, the source
transmits the codeword $\textbf{x}_{0}^{(2)}(w^{(2)}|w^{(1)},1)$
to the second relay, and having decoded the message $w^{(1)}$, the
first relay broadcasts the codewords
$\textbf{x}_{1}^{(2)}(w^{(1)}|1)$ and
$\textbf{u}_{1}^{(2)}(1)$ to the second relay and to the
destination. It should be noted that the destination cannot decode
the message $w^{(1)}$ at the end of this time slot; however, the
second relay decodes $w^{(1)}$ and $w^{(2)}$ messages. Using the binning
function, it finds the bin index of $w^{(1)}$ according to $s_{1}^{(1)}
= f_{Bin}^{(1)}(w^{(1)})$. In the third time slot, the source
transmits the codeword $\textbf{x}_{0}^{(1)}(w^{(3)}|w^{(2)},s_{1}^{(1)})$
to the first relay, and the second relay broadcasts the codewords
$\textbf{x}_{2}^{(1)}(w^{(2)}|s_{1}^{(1)})$ and
$\textbf{u}_{2}^{(1)}(s_{1}^{(1)})$ to the first relay and to the
destination.

Two types of decoding can be used at the
destination: successive decoding and backward decoding. Successive
decoding at the destination can be described as follows. At the end
of the $b$th time slot, the destination cannot decode the message
$w^{(b-1)}$; however, having decoded the bin index $s_{(b+1)\text{ mod } 2 +1}^{(b-2)}$ from the received vector of the $b$th time slot, it can
decode the message $w^{(b-2)}$ from $s_{(b+1)\text{ mod } 2 +1}^{(b-2)}$ and the received vector of the $(b-1)$th time slot. On the other hand, backward decoding can
be explained as follows. Having received the sequence of the $B+2$'th time slot, the final
destination starts decoding the intended messages. In the time slot
$B+2$, one of the relays transmits the dummy message $``1"$ along
with the bin index of the message $w^{(B)}$ to the destination. Having
received this bin index, the destination decodes it, and then
backwardly decodes messages $w^{(b)},~b=B,B-1,\cdots,1$ and their bin
indices. The following Theorem gives the achievable rate of the proposed scheme.

\begin{Theorem}\label{thm:Th1}
For the half-duplex parallel relay channel, assuming successive
relaying, the BME scheme achieves the rates $C_{BME_{succ}}^{low}$
and $C_{BME_{back}}^{low}$ using successive and backward decoding,
respectively:
\begin{IEEEeqnarray}{rl}
C_{BME_{succ}}^{low}& = R^{(1)} + R^{(2)} \leq\max_{0\leq t_{1},t_{2},t_{1}+t_{2}=1}\min\left( \right.\nonumber\\
&\left.
\min{\left(t_{1}I\left(X_{0}^{(1)};Y_{1}^{(1)}\mid X_{2}^{(1)},U_{2}^{(1)}\right), t_{2}I\left(X_{1}^{(2)}
;Y_{3}^{(2)}\mid U_{1}^{(2)}\right)+t_{1}I\left(U_{2}^{(1)};
Y_{3}^{(1)}\right)\right)}+\right.\nonumber\\
&\left.
\min{\left(t_{1}I\left(X_{2}^{(1)}
;Y_{3}^{(1)}\mid U_{2}^{(1)}\right)+t_{2}I\left(U_{1}^{(2)};
Y_{3}^{(2)}\right), t_{2}I\left(X_{0}^{(2)};Y_{2}^{(2)}\mid
X_{1}^{(2)},U_{1}^{(2)}\right)\right)},\right.\nonumber\\
&\left.
t_{1}I\left(X_{0}^{(1)},X_{2}^{(1)};Y_{1}^{(1)}\mid
U_{2}^{(1)}\right), t_{2}I\left(X_{0}^{(2)},X_{1}^{(2)}
;Y_{2}^{(2)}\mid
U_{1}^{(2)}\right)\right).\label{eq:succ}
\end{IEEEeqnarray}
with probabilities
\begin{IEEEeqnarray}{rl}
&p(x_{0}^{(1)},x_{2}^{(1)},u_{2}^{(1)})=p(u_{2}^{(1)})p(x_{2}^{(1)}|u_{2}^{(1)})p(x_{0}^{(1)}|x_{2}^{(1)},u_2^{(1)}),\nonumber\\
&p(x_{0}^{(2)},x_{1}^{(2)},u_{1}^{(2)})=p(u_{1}^{(2)})p(x_{1}^{(2)}|u_{1}^{(2)})p(x_{0}^{(2)}|x_{1}^{(2)},u_1^{(2)}),\nonumber\\
&p(x_{2}^{(1)},u_{2}^{(1)})=p(u_{2}^{(1)})p(x_{2}^{(1)}|u_{2}^{(1)}),\nonumber\\
&p(x_{1}^{(2)},u_{1}^{(2)})=p(u_{1}^{(2)})p(x_{1}^{(2)}|u_{1}^{(2)}).\nonumber
\end{IEEEeqnarray}
\begin{IEEEeqnarray}{rl}
C_{BME_{back}}^{low}& = R^{(1)} + R^{(2)} \leq\nonumber\\
&\max_{0\leq t_{1},t_{2},t_{1}+t_{2}=1}\min\left(
t_{1}I\left(X_{0}^{(1)},X_{2}^{(1)};Y_{1}^{(1)}\right),t_{2}I\left(X_{0}^{(2)},X_{1}^{(2)}
;Y_{2}^{(2)}\right),\right.\nonumber\\
&\left.t_{1}I\left(X_{0}^{(1)};Y_{1}^{(1)}\mid X_{2}^{(1)}\right)+t_{2}I\left(X_{0}^{(2)};Y_{2}^{(2)}\mid
X_{1}^{(2)}\right),\right.\nonumber\\
&\left.t_1I\left(X_{2}^{(1)};Y_{3}^{(1)}\right)+
t_2I\left(X_{1}^{(2)};Y_{3}^{(2)}\right)\right).\label{eq:Back1}
\end{IEEEeqnarray}
with probabilities
\begin{IEEEeqnarray}{rl}
&p(x_{0}^{(1)},x_{2}^{(1)})=p(x_{2}^{(1)})p(x_{0}^{(1)}|x_{2}^{(1)}),\nonumber\\
&p(x_{0}^{(2)},x_{1}^{(2)})=p(x_{1}^{(2)})p(x_{0}^{(2)}|x_{1}^{(2)}).\nonumber
\end{IEEEeqnarray}
\end{Theorem}
\begin{proof}
See Appendix B.
\end{proof}

Now, the following set of propositions and corollaries investigate the Non-Cooperative and Cooperative schemes and compare them with each other.
\begin{Proposition}
The BME with backward decoding achieves a better rate than the one with successive decoding, i.e., $C_{BME_{back}}^{low}\geq C_{BME_{succ}}^{low}$.
\end{Proposition}
\begin{proof}
For the first term of minimization (\ref{eq:succ}), we have
\begin{IEEEeqnarray}{rl}
&\min{\left(t_{1}I\left(X_{0}^{(1)};Y_{1}^{(1)}\mid X_{2}^{(1)},U_{2}^{(1)}\right), t_{2}I\left(X_{1}^{(2)}
;Y_{3}^{(2)}\mid U_{1}^{(2)}\right)+t_{1}I\left(U_{2}^{(1)};
Y_{3}^{(1)}\right)\right)}+\nonumber\\
&\min{\left(t_{1}I\left(X_{2}^{(1)}
;Y_{3}^{(1)}\mid U_{2}^{(1)}\right)+t_{2}I\left(U_{1}^{(2)};
Y_{3}^{(2)}\right), t_{2}I\left(X_{0}^{(2)};Y_{2}^{(2)}\mid
X_{1}^{(2)},U_{1}^{(2)}\right)\right)}\leq \nonumber\\
&\min \left(t_1I\left(X_{0}^{(1)};Y_{1}^{(1)}\mid X_{2}^{(1)},U_{2}^{(1)}\right)+
t_2I\left(X_{0}^{(2)};Y_{2}^{(2)}\mid
X_{1}^{(2)},U_{1}^{(2)}\right),\right.\nonumber\\
&\left.t_1I\left(X_2^{(1)},U_2^{(1)};Y_3^{(1)}\right) + t_2I\left(X_1^{(2)},U_1^{(2)};Y_3^{(2)}\right)\right).\label{eq:R_back_succ}
\end{IEEEeqnarray}
Let us focus on $t_1I\left(X_{0}^{(1)};Y_{1}^{(1)}\mid X_{2}^{(1)},U_{2}^{(1)}\right)+
t_2I\left(X_{0}^{(2)};Y_{2}^{(2)}\mid
X_{1}^{(2)},U_{1}^{(2)}\right)$:
\begin{IEEEeqnarray}{rl}
&t_1I\left(X_{0}^{(1)};Y_{1}^{(1)}\mid X_{2}^{(1)},U_{2}^{(1)}\right)+
t_2I\left(X_{0}^{(2)};Y_{2}^{(2)}\mid
X_{1}^{(2)},U_{1}^{(2)}\right)\stackrel{(a)}{=}\nonumber\\
&t_1H\left(Y_{1}^{(1)}\mid X_{2}^{(1)},U_{2}^{(1)}\right)-t_1H\left(Y_{1}^{(1)}\mid X_{0}^{(1)},X_{2}^{(1)}\right)+\nonumber\\
&t_2H\left(Y_{2}^{(2)}\mid X_{1}^{(2)},U_{1}^{(2)}\right)-t_2H\left(Y_{2}^{(2)}\mid X_{0}^{(2)},X_{1}^{(2)}\right)\stackrel{(b)}{\leq}\nonumber\\
&t_1H\left(Y_{1}^{(1)}\mid X_{2}^{(1)}\right)-t_1H\left(Y_{1}^{(1)}\mid X_{0}^{(1)},X_{2}^{(1)}\right)+\nonumber\\
&t_2H\left(Y_{2}^{(2)}\mid X_{1}^{(2)}\right)-t_2H\left(Y_{2}^{(2)}\mid X_{0}^{(2)},X_{1}^{(2)}\right)\stackrel{(c)}{=}\nonumber\\
&t_1I\left(X_{0}^{(1)};Y_{1}^{(1)}\mid X_{2}^{(1)}\right)+
t_2I\left(X_{0}^{(2)};Y_{2}^{(2)}\mid
X_{1}^{(2)}\right).
\end{IEEEeqnarray}
$(a)$ and $(c)$ follow from the definition of mutual information, the fact that  $U_2^{(1)}\longrightarrow
\left(X_0^{(1)},X_2^{(1)}\right)$ $\longrightarrow Y_1^{(1)}$ and $U_1^{(2)}\longrightarrow
\left(X_0^{(2)},X_1^{(2)}\right)\longrightarrow Y_2^{(2)}$ form Markov chain, and $(b)$ follows from the fact that conditioning reduces entropy. Inequality $(b)$ becomes equality if $p(x_{0}^{(1)},x_{2}^{(1)},u_2^{(1)})=p(u_{2}^{(1)})p(x_{2}^{(1)})p(x_{0}^{(1)}|x_{2}^{(1)})$ and $p(x_{0}^{(2)},x_{1}^{(2)},u_1^{(2)})=p(u_{1}^{(2)})p(x_{1}^{(2)})p(x_{0}^{(2)}|x_{1}^{(2)})$ . Using the similar argument for $t_1I\left(X_2^{(1)},U_2^{(1)};Y_3^{(1)}\right) + t_2I\left(X_1^{(2)},U_1^{(2)};Y_3^{(2)}\right)$,\\ $t_{1}I\left(X_{0}^{(1)},X_{2}^{(1)};Y_{1}^{(1)}\mid
U_{2}^{(1)}\right)$, and $t_{2}I\left(X_{0}^{(2)},X_{1}^{(2)};Y_{2}^{(2)}\mid
U_{1}^{(2)}\right)$ in (\ref{eq:succ}) and (\ref{eq:R_back_succ}), and the fact $U_2^{(1)}\longrightarrow X_2^{(1)}\longrightarrow Y_3^{(1)}$, $U_1^{(2)}\longrightarrow X_1^{(2)}\longrightarrow Y_3^{(2)}$, $U_2^{(1)}\longrightarrow
\left(X_0^{(1)},X_2^{(1)}\right)\longrightarrow Y_1^{(1)}$, $U_1^{(2)}\longrightarrow
\left(X_0^{(2)},X_1^{(2)}\right)\longrightarrow Y_2^{(2)}$ form Markov chain, and Appendix B, along with comparing $C_{BM_{succ}}^{low}$ and $C_{BM_{back}}^{low}$ in Theorem \ref{thm:Th1}, we have $C_{BM_{back}}^{low}\geq C_{BM_{succ}}^{low}$.
\end{proof}
From Theorem \ref{thm:Th1}, we have the following
corollary for the Gaussian case.

\begin{corollary}\label{cor:cor1}
For the half-duplex Gaussian parallel relay channel, assuming
successive relaying protocol with power constraint at the source and
each relay, BME achieves the following rates
\begin{eqnarray}
C_{BME_{succ}}^{low}&=&\max~\min\left(C_{BME_1}^{low}+C_{BME_2}^{low},\right.\nonumber\\
&&\left.
t_{1}C\left(\frac{h_{01}^2P_{0}^{(1)}+h_{12}^2\theta_{2}P_2
+2h_{01}h_{12}\sqrt{\bar{\alpha}_{1}\theta_{2}P_{0}^{(1)}P_2}}{t_{1}}\right),\right.\nonumber\\
&&\left.
t_{2}C\left(\frac{h_{02}^2P_{0}^{(2)}+h_{12}^2\theta_{1}P_1
+2h_{02}h_{12}\sqrt{\bar{\alpha}_{2}\theta_{1}P_{0}^{(2)}P_1}}{t_{2}}\right)\right).\label{eq:R_BM_1}
\end{eqnarray}
\begin{eqnarray}
C_{BME_{back}}^{low}&=&\max~\min\left(t_1C\left(\frac{h_{01}^2P_0^{(1)}
+h_{12}^{2}P_2+2h_{01}h_{12}\sqrt{\bar{\beta}_1P_0^{(1)}P_2}}{t_1}\right),\right.\nonumber\\
&&\left.
t_2C\left(\frac{h_{02}^2P_0^{(2)}
+h_{12}^{2}P_1+2h_{02}h_{12}\sqrt{\bar{\beta}_2P_0^{(2)}P_1}}{t_2}\right),\right.\nonumber\\
&&\left.
t_1C\left(\frac{h_{01}^{2}\beta_1P_0^{(1)}}{t_1}\right)+
t_2C\left(\frac{h_{02}^{2}\beta_2P_0^{(2)}}{t_2}\right),\right.\nonumber\\
&&\left.t_1C\left(\frac{h_{23}^{2}P_2}{t_1}\right)+
t_2C\left(\frac{h_{13}^{2}P_1}{t_2}\right)\right).\label{eq:R_BM_b}
\end{eqnarray}
\begin{IEEEeqnarray}{rl}
&\text{subject to:}\nonumber\\
&~~~~~C_{BME_1}^{low} = \min\left(t_{1}C\left(\frac{h_{01}^2\alpha_{1}P_{0}^{(1)}}{t_{1}}\right),
t_{1}C\left(\frac{h_{23}^2\bar{\theta}_{2}P_2}{h_{23}^2\theta_{2}P_2+t_{1}}\right)
+t_{2}C\left(\frac{h_{13}^2\theta_{1}P_1}{t_{2}}\right)\right),\label{eq:R_BM_2}\\
&~~~~~C_{BME_2}^{low} = \min\left(t_{2}C\left(\frac{h_{02}^2\alpha_{2}P_{0}^{(2)}}{t_{2}}\right),
t_{2}C\left(\frac{h_{13}^2\bar{\theta}_{1}P_1}{h_{13}^2\theta_{1}P_1+t_{2}}\right)
+t_{1}C\left(\frac{h_{23}^2\theta_{2}P_2}{t_{1}}\right)\right),\label{eq:R_BM_3}\\
&~~~~~P_{0}^{(1)} + P_{0}^{(2)} = P_0,\nonumber\\
&~~~~~t_{1} + t_{2} = 1,\nonumber\\
&~~~~~0\leq \alpha_{1},\alpha_{2} \leq1,\nonumber\\
&~~~~~0\leq \beta_{1},\beta_{2} \leq1,\nonumber\\
&~~~~~0\leq \theta_{1},\theta_{2} \leq1.\nonumber
\end{IEEEeqnarray}

where $\bar{\theta}_{i} = 1 - \theta_{i}$, $\bar{\alpha}_{i} = 1 - \alpha_{i}$, and $\bar{\beta}_{i} = 1 - \beta_{i}$ for $i = 1, 2$.\\
\end{corollary}
\begin{proof}
Let $V_{0}^{(1)}\sim\mathcal{N}(0,\alpha_{1}P_{0}^{(1)}), V_{0}^{(2)}\sim\mathcal{N}
(0,\alpha_{2}P_{0}^{(2)}),$ $V_{2}^{(1)}\sim\mathcal{N}(0,\theta_{2}P_2),
V_{1}^{(2)}\sim\mathcal{N}(0,\theta_{1}P_1)$, $ U_{2}^{(1)}\sim\mathcal{N}(0,\bar{\theta}_{2}P_2)$
and $U_{1}^{(2)}\sim\mathcal{N}(0,\bar{\theta}_{1}P_1)$, which are independent of each other.

Letting $X_{0}^{(1)} = V_{0}^{(1)} + \sqrt{\frac{\bar{\alpha}_{1}
P_{0}^{(1)}}{\theta_{2}P_2}}V_{2}^{(1)}, X_{0}^{(2)} =
V_{0}^{(2)} + \sqrt{\frac{\bar{\alpha}_{2}
P_{0}^{(2)}}{\theta_{1}P_1}} V_{1}^{(2)},X_{2}^{(1)} = V_{2}^{(1)} +
U_{2}^{(1)}$, $X_{1}^{(2)} = V_{1}^{(2)} +
U_{1}^{(2)}$ and using the result in the expression for the
achievable rate obtained in Theorem~\ref{thm:Th2}, we obtain $C_{{BME}_{succ}}^{low}$ for the Gaussian case, as given in~\cite{Yong} and (\ref{eq:R_BM_1}), (\ref{eq:R_BM_2}), and (\ref{eq:R_BM_3}), respectively.

For backward decoding, let $V_{0}^{(1)}\sim\mathcal{N}(0,\beta_{1}P_{0}^{(1)}), V_{0}^{(2)}\sim\mathcal{N}
(0,\beta_{2}P_{0}^{(2)}),$ $X_{2}^{(1)}\sim\mathcal{N}(0,P_2)$, and
$X_{1}^{(2)}\sim\mathcal{N}(0,P_1)$, which are independent of each other. By setting $X_{0}^{(1)} = V_{0}^{(1)} + \sqrt{\frac{\bar{\beta}_{1}
P_{0}^{(1)}}{P_2}}X_{2}^{(1)}$, $X_{0}^{(2)} =
V_{0}^{(2)} + \sqrt{\frac{\bar{\beta}_{2}
P_{0}^{(2)}}{P_1}} X_{1}^{(2)}$ and using the result in the expression for the
achievable rate obtained in Theorem~\ref{thm:Th2}, we obtain $C_{{BME}_{back}}^{low}$ for the Gaussian case, as given in (\ref{eq:R_BM_b}).
\end{proof}
\begin{Proposition}
In symmetric scenarios, where $h_{01}=h_{02}$, $h_{13}=h_{23}$, and $P_1=P_2$, Non-Cooperative DPC scheme outperforms Cooperative BME scheme, i.e. $C_{BME_{back}}^{low}\leq C_{DPC}^{low}$.
\end{Proposition}
\begin{proof}
Due to the symmetric assumption, we have $t_1=t_2=\frac{1}{2}$, $P_{0}^{(1)}=P_{0}^{(2)}=\frac{P_0}{2}$, and $\beta_1=\beta_2=\frac{1}{2}$. Hence, from~(\ref{eq:R_BM_b}), we have
\begin{eqnarray}
C_{BME_{back}}^{low}\leq\min\left(C\left(\frac{h_{01}^2P_0}{2}\right),
C\left(2h_{13}^2P_1\right)\right).\label{eq:R_BM_Sym}
\end{eqnarray}
And also $C_{DPC}^{low}$ in~(\ref{eq:R_DPC_1}) becomes
\begin{eqnarray}
C_{DPC}^{low}=\min\left(C\left(h_{01}^2P_0\right),
\frac{1}{2}C\left(h_{01}^2P_0\right)+\frac{1}{2}C\left(2h_{13}^2P_1\right),
C\left(2h_{13}^2P_1\right)\right).\label{eq:R_DPC_Sym}
\end{eqnarray}
Comparing~(\ref{eq:R_BM_Sym}) and~(\ref{eq:R_DPC_Sym}), we have $C_{BME_{back}}^{low}\leq C_{DPC}^{low}$.
\end{proof}

According to the discussion in Appendix B, $r_{Bin}^{(1)}=0$ or $r_{Bin}^{(2)}=0$. In other words, in the Cooperative BME scheme based on backward decoding, at most one relay is necessary to use binning function for the message it receives from another, and the other relay is not necessary to cooperate with this relay. Therefore, we propose a composite BME-DPC scheme. In this scheme, one of the relays decodes the other relay's message. Having decoded that, it then uses the binning function to cooperate with the other relay. On the other hand, using the Gelfand- Pinsker's result the source cancels the interference due to one relay on the other. Hence, we have the following Theorem.
\begin{Theorem}\label{thm:Th21}
The composite BME-DPC scheme, achieves the following rate:
\begin{IEEEeqnarray}{rl}
C_{BME-DPC}^{low}=&\max_{0\leq t_1,t_2,t_1+t_2=1} \min \left(t_{1}I\left(X_{0}^{(1)},X_{2}^{(1)};Y_{1}^{(1)}\right),
t_1I\left(X_0^{(1)};Y_1^{(1)}\mid X_2^{(1)}\right)+\right.\nonumber\\
&\left.t_2\left(I\left(U_0^{(2)};Y_2^{(2)}\right)-I\left(U_0^{(2)};X_1^{(2)}\right)\right),
t_1I\left(X_2^{(1)};Y_3^{(1)}\right)+t_2I\left(X_1^{(2)};Y_3^{(2)}\right),\right.\nonumber\\
&\left. t_2\left(I\left(U_0^{(2)};Y_2^{(2)}\right)-I\left(U_0^{(2)};X_1^{(2)}\right)\right) + t_2I\left(X_1^{(2)};Y_3^{(2)}\right)\right).
\end{IEEEeqnarray}
\end{Theorem}
\begin{proof}
Assuming $r_{Bin}^{(1)}=0$, and using Theorem~\ref{thm:Th2} and Theorem~\ref{thm:Th1} along with a similar argument as in Appendix B, Theorem~\ref{thm:Th21} is immediate.
\end{proof}

\begin{corollary}\label{cor:cor5}
For the Gaussian case, the composite BME-DPC scheme achieves the following rate $C_{BME-DPC}^{low}$. Furthermore, $C_{BME-DPC}^{low}\geq C_{BME_{back}}^{low}$. In other words, the composite BME-DPC scheme always achieves a better rate than the BME scheme for the Gaussian scenario.
\begin{IEEEeqnarray}{rl}
C_{BME-DPC}^{low} = &R^{(1)} + R^{(2)} \leq\nonumber\\
&\max \min\left(
t_{1}C\left(\frac{h_{01}^2P_{0}^{(1)}+h_{12}^2 P_2
+2h_{01}h_{12}\sqrt{\bar{\alpha}P_{0}^{(1)} P_2}}{t_{1}}\right),\right.\nonumber\\
&\left.t_1C\left(\frac{h_{01}^2\alpha P_0^{(1)}}{t_1}\right)+
t_2C\left(\frac{h_{02}^2P_0^{(2)}}{t_2}\right),\right.\nonumber\\
&\left.t_1C\left(\frac{h_{23}^2P_2}{t_1}\right)+
t_2C\left(\frac{h_{13}^2P_1}{t_2}\right),t_2C\left(\frac{h_{02}^2P_0^{(2)}}{t_2}\right)+
t_2C\left(\frac{h_{13}^2P_1}{t_2}\right)\right).\label{eq:nco_cop}\\
&\text{subject to:}\nonumber\\
&~~~~~P_0^{(1)} + P_0^{(2)} = P_0,\nonumber\\
&~~~~~t_1 + t_2 = 1,\nonumber\\
&~~~~~0\leq t_{1}, t_{2}, P_0^{(1)}, P_0^{(2)},\nonumber\\
&~~~~~0\leq \alpha \leq 1.\nonumber
\end{IEEEeqnarray}
where $\bar{\alpha} = 1 - \alpha$.
\end{corollary}
\begin{proof}
As in Theorem~\ref{thm:Th21}, we assume that $r_{Bin}^{(1)}=0$. Now, we show that every rate pairs $\left(R^{(1)},R^{(2)}\right)$ satisfying (\ref{eq:constr1})-(\ref{eq:constr2}) satisfy (\ref{eq:nco_cop}). After specializing (\ref{eq:constr1})-(\ref{eq:constr2}) for the Gaussian case and comparing with (\ref{eq:nco_cop}), one observes that the second term in minimization (\ref{eq:constr1}) does not exist. Substituting $r_{Bin}^{(1)}=0$ in (\ref{eq:constr3})-(\ref{eq:constr2}), one can obtain the other three corresponding terms. Comparing those terms with (\ref{eq:nco_cop}), it can be readily seen that $C_{BME-DPC}^{low}\geq C_{BME_{back}}^{low}$.
\end{proof}
\begin{Remark}
Assuming $r_{Bin}^{(1)} = 0$, as in Theorem \ref{thm:Th21} and corollary \ref{cor:cor5}, the destination jointly decodes the current message and the bin index of the next message at the end of even time slots and then it can decode the next message at the end of odd time slots. Therefore, using backward decoding is not necessary in the BME-DPC scheme.
\end{Remark}
\subsection{Simultaneous Relaying Protocol}\label{sec:SimRate}
Figure~\ref{fig:Ch3_6.eps} shows simultaneous relaying protocol. In simultaneous relaying, in one time slot of duration $t_3$ the source transmits its signal simultaneously to the two relays. In the next time slot of duration $t_4$, two relays transmit their signal coherently to the destination. Hence, in this protocol, $t_1 = t_2 = 0$ and our system model follows from (\ref{eq:Sys3}) and (\ref{eq:Sys4}).
\begin{figure}[tp]
\centering
\includegraphics[angle=0,scale=.75]{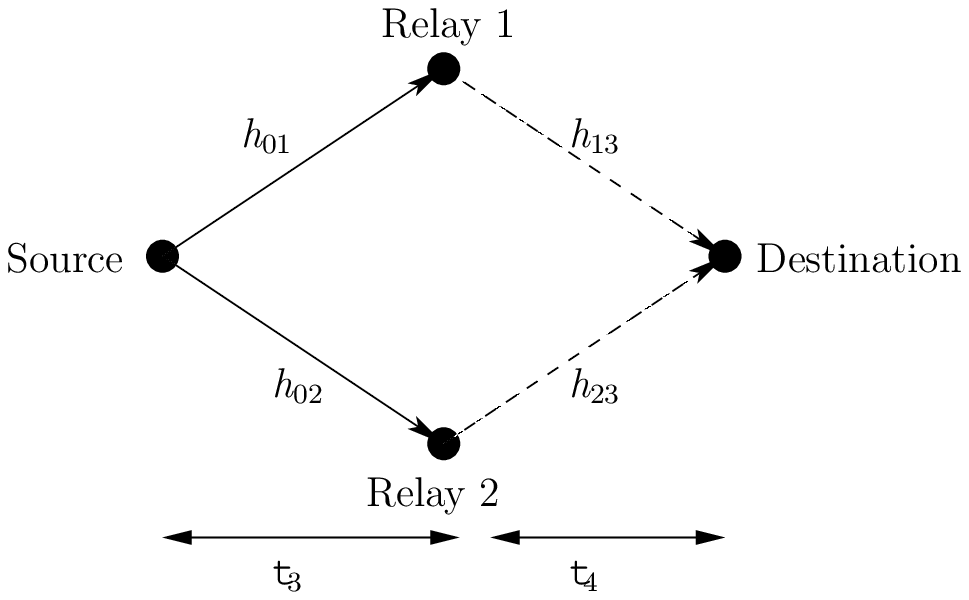}\\
\caption{\small{Simultaneous relaying protocol for two relays.}}\label{fig:Ch3_6.eps}
\end{figure}
\subsubsection{Dynamic Decode-and-Forward (DDF)}
 In DDF scheme each relay decodes the transmitted message from the source in time slot $t_3$ (Broadcast (BC) State), and forwards its re-encoded version in time slot $t_4$ (Multiple Access (MAC) State). The following Theorem gives the achievable rate of the DDF scheme for the general discrete memoryless channels.
 \begin{Theorem}
 For the half-duplex parallel relay channel, assuming simultaneous relaying and the fact that what the second relay receives is a degraded version of what the first relay receives, the following rate $C_{DDF}^{low}$ is achievable:
\begin{IEEEeqnarray}{rl}
C_{DDF}^{low} = &\max_{0\leq t_{3},t_{4},t_{3}+t_{4}=1} R_p + R_c,\label{eq:DDF1}\\
&\text{subject to:}\nonumber\\
&~~~~~R_p\leq \min\left(t_3I(X_{0}^{(3)};Y_{1}^{(3)}\mid U_{0}^{(3)}),
t_4I(X_{1}^{(4)};Y_{3}^{(4)}\mid X_{2}^{(4)})\right),\label{eq:DDF2}\\
&~~~~~R_c\leq t_3I(U_{0}^{(3)};Y_{2}^{(3)}),\label{eq:DDF3}\\
&~~~~~R_{p} + R_{c}\leq t_4I(X_{1}^{(4)},X_{2}^{(4)};Y_{3}^{(4)}).\label{eq:DDF4}
\end{IEEEeqnarray}
with probabilities:
\begin{IEEEeqnarray}{rl}
&p(u_{0}^{(3)},x_{0}^{(3)}) = p(u_{0}^{(3)})p(x_{0}^{(3)}|u_{0}^{(3)}),\nonumber\\
&p(x_{1}^{(4)},x_{2}^{(4)}) = p(x_{1}^{(4)})p(x_{2}^{(4)}|x_{1}^{(4)}).\nonumber
\end{IEEEeqnarray}
 \end{Theorem}
 \begin{proof}
 The achievable rate of DDF is equal to $C_{DDF}^{low} = R_{p} + R_{c}$, where ($R_{p}$, $R_{c}$) should be both in the capacity region of BC (corresponding to the BC state) and MAC (corresponding to the MAC state).
 Applying the superposition coding of the degraded BC~\cite{CoverThomas} the following rates are achievable for the first hop:
 \begin{IEEEeqnarray}{rl}\label{eq:SimRate2}
&R_{p}\leq t_3I(X_{0}^{(3)};Y_{1}^{(3)}\mid U_{0}^{(3)}),\nonumber\\
&R_{c}\leq t_3I(U_{0}^{(3)};Y_{2}^{(3)}).
\end{IEEEeqnarray}
with probability $p(u_{0}^{(3)},x_{0}^{(3)}) = p(u_{0}^{(3)})p(x_{0}^{(3)}|u_{0}^{(3)})$.

And using the superposition coding of the extended MAC (See \cite{Han},
\cite{Prelov}) the following rates are achievable for the second hop:
\begin{IEEEeqnarray}{rl}\label{eq:SimRate3}
&R_{p}\leq t_4I(X_{1}^{(4)};Y_{3}^{(4)}\mid X_{2}^{(4)}),\nonumber\\
&R_{p} + R_{c}\leq t_4I(X_{1}^{(4)},X_{2}^{(4)};Y_{3}^{(4)}).
\end{IEEEeqnarray}
with probability $p(x_{1}^{(4)},x_{2}^{(4)}) = p(x_{1}^{(4)})p(x_{2}^{(4)}|x_{1}^{(4)})$.
\end{proof}

In the Gaussian case (assuming $h_{01}\geq h_{02}$), the source splits its total
available power $P_{0}$ to $P_{0,p}^{(3)}$ and $P_{0,c}^{(3)}$ associated with
the \emph{``Private"} and the \emph{``Common"} messages,
respectively. Letting $X_{0}^{(3)}\sim\mathcal{N}\left(0,P_{0}\right)$, $U_{0}^{(3)}\sim\mathcal{N}\left(0,P_{0,c}^{(3)}\right)$, and $X_{1}^{(4)}\sim\mathcal{N}\left(0,P_{1}\right)$, assuming that relay 1 and relay 2 transmit their codewords associated with the common message with $\mathcal{N}\left(0,P_{1,c}^{(4)}\right)$ and $\mathcal{N}\left(0,P_{2}\right)$, and using
(\ref{eq:SimRate2}) and (\ref{eq:SimRate3}) we have the following corollary.
\newpage
\begin{corollary}\label{cor:cor4}
For the half-duplex Gaussian parallel relay channel, assuming
simultaneous relaying protocol with power constraints at the source and
at each relay, DDF achieves the following rate
\begin{IEEEeqnarray}{rl}
C_{DDF}^{low}= &R_{p} + R_{c},\label{eq:Sim1_1}\\
\text{subject to:}
&~~~~~R_{p}\leq \min \left( t_3 C\left(\frac{h_{01}^{2}P_{0,p}^{(3)}}{t_3}\right), t_4 C\left(\frac{h_{13}^{2}P_{1,p}^{(4)}}{t_4}\right) \right),\nonumber\\
&~~~~~R_{c}\leq t_3 C\left(\frac{h_{02}^{2}P_{0,c}^{(3)}}{t_3+h_{02}^{2}P_{0,p}^{(3)}}\right),\nonumber\\
&~~~~~R_{p} + R_{c}\leq t_4 C\left(\frac{h_{13}^{2}P_{1,p}^{(4)}+\left(h_{13}\sqrt{P_{1,c}^{(4)}}+ h_{23}\sqrt{P_2}\right)^{2}}{t_4}\right),\nonumber\\
&~~~~~P_{0,p}^{(3)}+P_{0,c}^{(3)}=P_0,~ P_{1,p}^{(4)}+P_{1,c}^{(4)}=P_1,~ t_3+t_4=1,\nonumber\\
&~~~~~0\leq t_3,~ t_4,~ P_{0,p}^{(3)},~ P_{0,c}^{(3)},~ P_{1,p}^{(4)},~ P_{1,c}^{(4)}.\nonumber
\end{IEEEeqnarray}
\end{corollary}

Interestingly, successive decoding at the destination does not degrade the performance of DDF scheme in the Gaussian scenario as shown in the following Proposition.
\begin{Proposition}\label{pro:pro1}
The rate of DDF scheme is achievable by successive decoding of the
common and private messages at the destination.
\end{Proposition}
\begin{proof}
Consider the sum rate for both the common message and the private
message for the extended multiple access channel from relays to the
destination,
\begin{equation}\label{eq:R_P_C}
R_{p} + R_{c}\leq t_4 C\left(\frac{h_{13}^{2}P_{1,p}^{(4)}+(h_{13}\sqrt{P_{1,c}^{(4)}}+ h_{23}\sqrt{P_2})^{2}}{t_4}\right).
\end{equation}
It can be readily verified that subject to the constraint $P_{1,p}^{(4)} + P_{1,c}^{(4)} = P_1$, the right-hand side of (\ref{eq:R_P_C}) is a decreasing
function of $P_{1,p}^{(4)}$ or equivalently an increasing function of $P_{1,c}^{(4)}$.
Now, let us equate $R_{p}$ in~(\ref{eq:R_P_C}) with the private rate
$\acute{R}_{p}$ of another MAC which is achieved by successive
decoding of common and private messages. Therefore, we have
\begin{eqnarray}
R_{p} = \acute{R}_{p} = t_4 C\left(\frac{h_{13}^{2}\acute{P}_{1,p}^{(4)}}{t_4}\right)\leq t_4 C\left(\frac{h_{13}^{2}P_{1,p}^{(4)}}{t_4}\right).\label{eq:Rp}
\end{eqnarray}
According to (\ref{eq:Rp}), we have (See Fig. \ref{fig:R_P_C.eps})
\begin{eqnarray}
\acute{P}_{1,p}^{(4)}&\leq& P_{1,p}^{(4)}\Longrightarrow \nonumber\\
R_{p} + R_{c} &\leq& \acute{R}_{p} + \acute{R}_{c},\nonumber\\
R_{c}&\leq& \acute{R}_{c}.\nonumber
\end{eqnarray}
Hence, $(R_p,R_c)$ lies in the corner point of the extended MAC with parameters $(\acute{P}_{1,p}^{(4)},\acute{P}_{1,c}^{(4)})$, i.e. successive decoding of common and private messages achieves the DF rate.
\end{proof}

\begin{figure}[hbtp]
\centering
\includegraphics[angle=0,scale=.75]{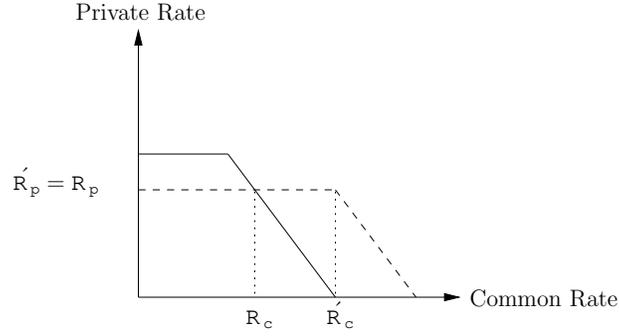}\\
\caption{\small{The order of decoding \emph{``Common"} and \emph{``Private"} messages.}}\label{fig:R_P_C.eps}
\end{figure}
\subsection{Simultaneous-Successive Relaying Protocol based on Dirty paper coding (SSRD)}
\begin{figure}[hbtp]
\centering
\includegraphics[angle=0,scale=.75]{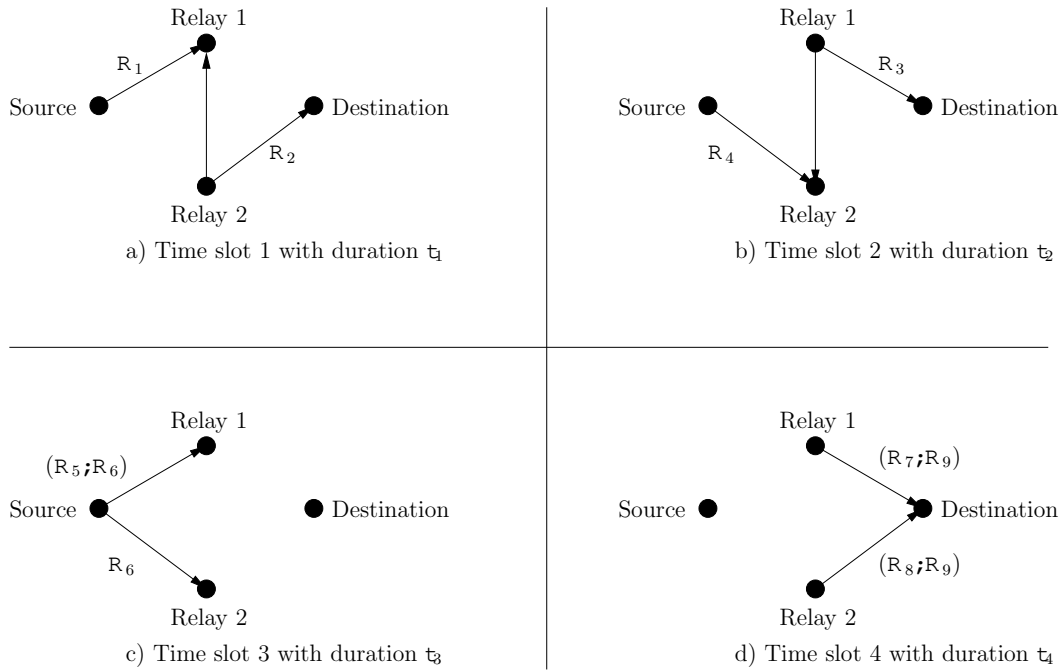}\\
\caption{\small{SSRD Scheme for the Half-Duplex Parallel Relay Channel.}}\label{fig:GeneralAchm.eps}
\end{figure}
In this section, we propose an achievable rate for the half-duplex parallel relay channel. Our achievable scheme is based on the combination of the successive relaying protocol based on DPC scheme and simultaneous relaying protocol based on DDF (SSRD scheme).
Hence, we have the following Theorem.

\begin{Theorem}\label{thm:ThSSRD}
Considering Fig.~\ref{fig:GeneralAchm.eps}, for the half-duplex parallel relay channel, SSRD scheme achieves the following rate $C_{SSRD}^{low}$:
\begin{IEEEeqnarray}{rl}\label{eq:Ach4}
C_{SSRD}^{low}=&\min\left(R_1+R_4+R_5+R_6,R_2+R_3+R_7+R_8+R_9\right),\\
&\text{subject to:}\nonumber\\
&~~~~~R_9\leq R_6,R_1+R_5\leq R_3+R_7,R_4\leq R_2 + R_8.\label{eq:Ach4_1}
\end{IEEEeqnarray}
\end{Theorem}
\begin{proof}
SSRD scheme is illustrated in Fig.~\ref{fig:GeneralAchm.eps}. As indicated in the figure, transmission is performed in 4 time slots. Relay 1 transmits its private message which was received in time slots $t_1$ and $t_3$ (corresponding to rates $R_1$ and $R_5$) in time slots $t_2$ and $t_4$ (corresponding to rates $R_3$ and $R_7$). On the other hand, relay 2 transmits its private message which has been received in time slot $t_2$ (corresponding to rate $R_4$) in time slots $t_1$ and $t_4$ (corresponding to rates $R_2$ and $R_8$). Furthermore, the two relays send the common message they have already received in time slot $t_3$ (corresponding to rate $R_6$) coherently in time slot $t_4$ (corresponding to rate $R_9$). As observed, here we consider the private rate for both relays in the MAC state, i.e. time slot $t_4$. This is due to the reason that relay 2 also receives the private message in time slot $t_2$. Hence, from the above description and Fig.~\ref{fig:GeneralAchm.eps}, we have
\begin{IEEEeqnarray}{rl}\label{eq:Ach4}
C_{SSRD}^{low}=&\min\left(R_1+R_4+R_5+R_6,R_2+R_3+R_7+R_8+R_9\right),\\
&\text{subject to:}\nonumber\\
&~~~~~R_9\leq R_6,~R_1+R_5\leq R_3+R_7,~R_4\leq R_2 + R_8.\label{eq:Ach4_1}
\end{IEEEeqnarray}
\end{proof}

Using corollaries~\ref{cor:cor2}, \ref{cor:cor4}, and Proposition~\ref{pro:pro1}, for the Gaussian case we have

\begin{IEEEeqnarray}{rl}
C_{SSRD}^{low} =& \min\left(t_{1}C\left(\frac{h_{01}^2P_{0}^{(1)}}{t_{1}}\right)+
t_{2}C\left(\frac{h_{02}^2P_{0}^{(2)}}{t_{2}}\right)+
t_{3}C\left(\frac{h_{01}^2P_{0,p}^{(3)}}{t_{3}}\right)+
t_3C\left(\frac{h_{02}^2P_{0,c}^{(3)}}{t_3+h_{02}^2P_{0,p}^{(3)}}\right),\right.\nonumber\\
&\left.t_{1}C\left(\frac{h_{23}^2P_{2}^{(1)}}{t_1}\right)
+t_{2}C\left(\frac{h_{13}^2P_{1}^{(2)}}{t_{2}}\right)+\right.\nonumber\\
&\left.t_4C\left(\frac{h_{13}^2P_{1,p}^{(4)}+h_{23}^2P_{2,p}^{(4)}}{t_4}\right)+
t_4C\left(\frac{\left(h_{13}\sqrt{P_{1,c}^{(4)}}+h_{23}\sqrt{P_{2,c}^{(4)}}\right)^2}
{t_4+h_{13}^2P_{1,p}^{(4)}+h_{23}^2P_{2,p}^{(4)}}\right)
\right),\\
&\text{subject to:}\nonumber\\
&~~~~~t_4C\left(\frac{\left(h_{13}\sqrt{P_{1,c}^{(4)}}+h_{23}\sqrt{P_{2,c}^{(4)}}\right)^2}
{t_4+h_{13}^2P_{1,p}^{(4)}+h_{23}^2P_{2,p}^{(4)}}\right)\leq
t_3C\left(\frac{h_{02}^2P_{0,c}^{(3)}}{t_3+h_{02}^2P_{0,p}^{(3)}}\right),\nonumber\\
&~~~~~t_{1}C\left(\frac{h_{01}^2P_{0}^{(1)}}{t_{1}}\right)
+t_{3}C\left(\frac{h_{01}^2P_{0,p}^{(3)}}{t_{3}}\right)\leq
t_{2}C\left(\frac{h_{13}^2P_{1}^{(2)}}{t_{2}}\right) + t_{4}C\left(\frac{h_{13}^2P_{1,p}^{(4)}}{t_{4}}\right),\nonumber\\
&~~~~~t_{2}C\left(\frac{h_{02}^2P_{0}^{(2)}}{t_{2}}\right)\leq
t_{1}C\left(\frac{h_{23}^2P_{2}^{(1)}}{t_1}\right)+
t_4C\left(\frac{h_{23}^2P_{2,p}^{(4)}}{t_4}\right),\nonumber\\
&~~~~~P_{0}^{(1)}+P_{0}^{(2)}+P_{0,p}^{(3)}+P_{0,c}^{(3)} = P_0,\nonumber\\
&~~~~~P_{1}^{(2)} + P_{1,p}^{(4)} + P_{1,c}^{(4)} = P_1,\nonumber\\
&~~~~~P_{2}^{(1)} + P_{2,p}^{(4)} + P_{2,c}^{(4)} = P_2,\nonumber\\
&~~~~~t_1 + t_2 + t_3 + t_4 = 1,\nonumber\\
&~~~~~0\leq t_1,~t_2,~t_3,~t_4,~P_0^{(1)},~P_{0}^{(2)},~P_{0,p}^{(3)},~P_{0,c}^{(3)},
~P_{1}^{(2)},~P_{1,p}^{(4)},~P_{1,c}^{(4)},~P_{2}^{(1)},~P_{2,p}^{(4)},~P_{2,c}^{(4)}.\nonumber
\end{IEEEeqnarray}

According to corollary~\ref{cor:cor5}, another combined simultaneous-successive relaying protocol based on BME is not necessary. However, a ``Simultaneous-Successive Relaying protocol based on BME-DPC", can be easily derived. Assuming the first relay decodes the second one's message, the achievable rate of this new scheme would be the same as $C_{SSRD}^{low}$. However, since the messages for the second relay are common, $R_8$ in the expression of the achievable rate is zero. Furthermore, the following constraints instead of (\ref{eq:Ach4_1}) should be satisfied:
\begin{IEEEeqnarray}{rl}\label{eq:Ach5}
&R_9\leq R_4 + R_6,~R_1+R_5\leq R_3+R_7,~R_1 + R_4\leq t_1I\left(X_0^{(1)},X_2^{(1)};Y_1^{(1)}\right).\label{eq:Ach5_1}
\end{IEEEeqnarray}

\section{Optimality Results}
In this section, an upper bound for the half-duplex parallel relay channel is derived and investigated. The authors in~\cite{Aazhang1}
proposed some upper bounds on the achievable rate for general
half-duplex multi-terminal networks. Here, we explain their results
briefly and apply them to our half-duplex parallel relay network.

Authors in~\cite{Aazhang1} define the concept of \emph{state} for a half-duplex network with $\emph{N}$ nodes. The state of the
network is a \emph{valid partitioning of its nodes into two sets of the
``sender nodes" and the ``receiver nodes" such that
there is no active link that arrives at a sender node}, and $\hat{t}_m$ is the
portion of the time that network is used in state $\emph{m}$ where
$\emph{m}\in\{1,2,\ldots,\emph{M}\}$. The following Theorem for the upper bound of the information flow from the subset $S_1$ to the subset $S_2$ of the nodes, where $S_1$ and $S_2$ are disjoint is
proved in~\cite{Aazhang1}.
\begin{Theorem}\label{thm:Th3}
For a general half-duplex network with $N$ nodes and a finite number of states, $M$, the maximum achievable
information rates $\{R^{ij}\}$ from a node set $S_{1}$ to a disjoint
node set $S_{2}$, $S_{1},S_{2}\subset\{0,1,\ldots,N-1\}$, is bounded by
\begin{equation}
\sum_{i\in S_{1},j\in S_{2}}{R^{ij}}\leq \sup_{p(x_{0}^{(m)},x_{2}^{(m)},\ldots,x_{N-1}^{(m)}),\hat{t}_{m}}\min_{S}{\sum_{m=1}^{M}{\hat{t}_{m}
I\left(X_{S}^{(m)};Y_{S}^{(m)}\mid X_{S^{c}}^{(m)}\right)}}.
\end{equation}
for some joint probability distribution
$p(x_{0}^{(m)},x_{2}^{(m)},\ldots,x_{N-1}^{(m)})$ when the minimization is
over all the sets $S\subset\{0,1,\ldots,N-1\}$ subject to $S\bigcap
S_{1} = S_{1}$, $S\bigcap S_{2} = \emptyset$ and the supremum is
over all the non-negative $\hat{t}_{m}$ subject to
$\sum_{i=1}^{M}{\hat{t}_{m}}=1$. Here, $x_S^{(m)}$, $y_S^{(m)}$, and $x_{S^{c}}^{(m)}$ denote the signals transmitted and received by nodes in set $S$, and transmitted by nodes in set $S^{c}$, during state $m$, respectively.
\end{Theorem}

From Theorem~\ref{thm:Th3}, the maximum achievable rate
$C^{low}$ is upper bounded as
\begin{IEEEeqnarray}{rl}\label{eq:UpperBound}
C^{low}\leq\emph{C}^{up}\triangleq\min&\left(\hat{t}_{1}I\left(X_0^{(1)};Y_1^{(1)}
\mid X_2^{(1)}\right)+\hat{t}_{2}I\left(X_0^{(2)}
;Y_2^{(2)}\mid X_1^{(2)}\right)+
\hat t_{3}I\left(X_0^{(3)};Y_1^{(3)},Y_2^{(3)}\right),\right.\nonumber\\
&\left.\hat t_{2}I\left(X_0^{(2)}, X_1^{(2)};Y_2^{(2)},Y_3^{(2)}\right)
+\hat t_{3}I\left(X_{0}^{(3)};Y_{2}^{(3)}\right) + \hat t_{4}I\left(X_{1}^{(4)};Y_{3}^{(4)}\mid X_{2}^{(4)}\right),\right.\nonumber\\
&\left.\hat t_{1}I\left(X_0^{(1)}, X_2^{(1)};Y_1^{(1)},Y_3^{(1)}\right)
+\hat t_{3}I\left(X_{0}^{(3)};Y_{1}^{(3)}\right)+ \hat t_{4}I\left(X_{2}^{(4)};Y_{3}^{(4)}\mid X_{1}^{(4)}\right),\right.\nonumber\\
&\left.\hat t_{1}I\left(X_2^{(1)};Y_3^{(1)}\right)
+\hat t_{2}I\left(X_1^{(2)};Y_3^{(2)}\right)
+\hat t_{4}I\left(X_1^{(4)},X_2^{(4)};Y_3^{(4)}\right)\right),\\
&\text{subject to}\nonumber\\
&~~~~~\hat t_1+\hat t_2+\hat t_3+\hat t_4=1.\nonumber
\end{IEEEeqnarray}
By setting $\hat{t}_3 = \hat{t}_4 = 0$ in (\ref{eq:UpperBound}), we obtain an upper bound on the successive relaying protocol which we call it \emph{successive cut-set bound} in the sequel.
\begin{Theorem}
In a degraded half-duplex parallel relay channel where the destination receives a degraded version of the received signals at relays, i.e. $X_2^{(1)}\longrightarrow Y_1^{(1)}\longrightarrow Y_3^{(1)}$ and $X_1^{(2)}\longrightarrow Y_2^{(2)}\longrightarrow Y_3^{(2)}$, BME based on backward decoding achieves the successive cut-set bound.
\end{Theorem}
\begin{proof}
Setting $\hat{t}_3 = \hat{t}_4 = 0$ in (\ref{eq:UpperBound}) and comparing the result with (\ref{eq:Back1}) the Theorem is proved.
\end{proof}
In high SNR scenarios, we have the following Theorem.
\begin{Theorem} \label{thm:Th4}
In high SNR scenarios, assuming non-zero source-relay and relay-destination links, when power available for the source and each
relay tends to infinity, time slots $\hat t_{3}$ and $\hat t_{4}$
in~(\ref{eq:UpperBound}) tend to zero as
$O\left(\frac{1}{\log P_0}\right)$. Furthermore, the
upper bound on the capacity of the half-duplex parallel relay channel
in high SNR scenarios is
\begin{eqnarray}
C^{up}=C_{DPC}^{low}+O\left(\frac{1}{\log P_0}\right).\nonumber
\end{eqnarray}
In other words, DPC achieves the capacity of a
half-duplex Gaussian parallel relay channel as SNR goes to infinity.
\end{Theorem}
\begin{proof}
Throughout the proof, we assume the power of the relays goes to
infinity as $P_1=\gamma_1 P_0,~P_2=\gamma_2 P_0$ where
$\gamma_1, \gamma_2$ are constants independent of the SNR. Substituting $X_{0}^{(1)}\sim\mathcal{N}(0,\hat{P}_{0}^{(1)})$, $X_{0}^{(2)}\sim\mathcal{N}(0,\hat{P}_{0}^{(2)})$, $X_{0}^{(3)}\sim\mathcal{N}(0,\hat{P}_{0}^{(3)})$, $X_{1}^{(2)}\sim\mathcal{N}(0,\hat{P}_{1}^{(2)})$, $X_{1}^{(4)}\sim\mathcal{N}(0,\hat{P}_{1}^{(4)})$, $X_{2}^{(1)}\sim\mathcal{N}(0,\hat{P}_{2}^{(1)})$, and $X_{2}^{(4)}\sim\mathcal{N}(0,\hat{P}_{2}^{(4)})$ in~(\ref{eq:UpperBound}), and assuming complete cooperation between the transmitting and receiving nodes for each cut in~(\ref{eq:UpperBound}), we have
\begin{IEEEeqnarray}{rl}\label{eq:GaussianUpperBound1}
C^{up}\leq&\min\left(\hat t_{1}C\left(\frac{h_{01}^{2}\hat{P}_{0}^{(1)}}{\hat t_{1}}\right)+
\hat t_{2}C\left(\frac{h_{02}^{2}\hat{P}_{0}^{(2)}}{\hat t_{2}}\right)+
\hat t_{3}C\left(\frac{(h_{01}^2 + h_{02}^2)\hat{P}_{0}^{(3)}}{\hat t_{3}}\right),\right.\nonumber\\
&\left.
\hat t_{2}C\left(\frac{h_{02}^2\hat{P}_{0}^{(2)}}{\hat t_{2}} + \frac{(h_{12}^2 + h_{13}^2)\hat{P}_{1}^{(2)}}{\hat t_{2}}
+ \frac{2h_{02}h_{12}\sqrt{\hat{P}_{0}^{(2)}\hat{P}_{1}^{(2)}}}{\hat t_{2}} +
\frac{h_{02}^2h_{13}^2\hat{P}_{0}^{(2)}\hat{P}_{1}^{(2)}}{\hat t_{2}^{2}}\right)+\right.\nonumber\\
&\left.\hat t_{3}C\left(\frac{h_{02}^{2}\hat{P}_{0}^{(3)}}{\hat t_{3}}\right)
+\hat t_{4}C\left(\frac{h_{13}^2\hat{P}_{1}^{(4)}}{\hat t_{4}}\right),\right.\nonumber\\
&\left.
\hat t_{1}C\left(\frac{h_{01}^2\hat{P}_{0}^{(1)}}{\hat t_{1}} + \frac{(h_{12}^2 + h_{23}^2)\hat{P}_{2}^{(1)}}{\hat t_{1}}
+ \frac{2h_{01}h_{12}\sqrt{\hat{P}_{0}^{(1)}\hat{P}_{2}^{(1)}}}{\hat t_{1}} +
\frac{h_{01}^2h_{23}^2\hat{P}_{0}^{(1)}\hat{P}_{2}^{(1)}}{\hat t_{1}^{2}}\right)+\right.\nonumber\\
&\left.\hat t_{3}C\left(\frac{h_{01}^{2}\hat{P}_{0}^{(3)}}{\hat t_{3}}\right)
+\hat t_{4}C\left(\frac{h_{23}^2\hat{P}_{2}^{(4)}}{\hat t_{4}}\right),\right.\nonumber\\
&\left.
\hat t_{1}C\left(\frac{h_{23}^2\hat{P}_{2}^{(1)}}{\hat t_{1}}\right)+
\hat t_{2}C\left(\frac{h_{13}^2\hat{P}_{1}^{(2)}}{\hat t_{2}}\right)+\right.\nonumber\\
&\left.\hat t_{4}C\left(\frac{h_{13}^2\hat{P}_{1}^{(4)}+h_{23}^2\hat{P}_{2}^{(4)}+2 h_{13}h_{23}\sqrt{\hat{P}_{1}^{(4)}\hat{P}_{2}^{(4)}}}{\hat t_{4}}\right)
\right).\\
&\text{subject to:}\nonumber\\
&~~~~~\hat P_{0}^{(1)}+\hat P_{0}^{(2)}+\hat P_{0}^{(3)}=P_0,\nonumber\\
&~~~~~\hat P_{1}^{(2)}+\hat P_{1}^{(4)}=P_1,\nonumber\\
&~~~~~\hat P_{2}^{(1)}+\hat P_{2}^{(4)}=P_2,\nonumber\\
&~~~~~\hat t_{1}+\hat t_{2}+\hat t_{3}+\hat t_{4}=1,\nonumber\\
&~~~~~0\leq \hat t_{1},~\hat t_{2},~\hat t_{3},~\hat t_{4},~\hat P_{0}^{(1)},~\hat P_{0}^{(2)},~\hat P_{0}^{(3)},~\hat P_{1}^{(2)},~\hat P_{1}^{(4)},~\hat P_{2}^{(1)},~\hat P_{2}^{(4)}.\nonumber
\end{IEEEeqnarray}
Furthermore, from corollary~\ref{cor:cor2}, the achievable rate of
the DPC scheme can be expressed as
\begin{IEEEeqnarray}{rl}\label{eq:R_DPC}
C_{DPC}^{low}=\min&\left(
t_{1}C\left(\frac{h_{01}^2{P_{0}^{(1)}}}{t_{1}}\right)+
t_{2}C\left(\frac{h_{02}^2P_{0}^{(2)}}{t_{2}}\right),\right.\nonumber\\
&\left. t_{2}C\left(\frac{h_{02}^2P_{0}^{(2)}}{t_{2}}\right) +
t_{2}C\left(\frac{h_{13}^2P_1}{t_{2}}\right),\right.\nonumber\\
&\left. t_{1}C\left(\frac{h_{01}^2P_{0}^{(1)}}{t_{1}}\right)+
 t_{1}C\left(\frac{h_{23}^2P_2}{t_{1}}\right),\right.\nonumber\\
&\left.t_{1}C\left(\frac{h_{23}^2P_2}{t_{1}}\right)+
t_{2}C\left(\frac{h_{13}^2P_1}{t_{2}}\right)\right).
\end{IEEEeqnarray}
By setting $P_{0}^{(1)} = P_{0}^{(2)} = \frac{P_0}{2}$ and $t_{1} =
t_{2} = 0.5$ in (\ref{eq:R_DPC}),
expression~(\ref{eq:R_DPC}) can be simplified as
\begin{eqnarray}
C_{DPC}^{low} \geq \frac{1}{2}\ln P_0 + c. \label{eq:T6_R_DPC}
\end{eqnarray}
where $c$ is some constant which depends on channel coefficients.
Knowing that the term corresponding to each cut-set in
(\ref{eq:GaussianUpperBound1}) for the optimum values of
$\hat t_1,\cdots,\hat t_4$ is indeed an upper-bound for $C_{DPC}^{low}$, and by
setting $\hat P_{0}^{(1)} = \hat P_{0}^{(2)} = \hat P_{0}^{(3)} = P_0$ in
(\ref{eq:GaussianUpperBound1}), we have the following inequality
between (\ref{eq:T6_R_DPC}) and the first cut of
(\ref{eq:GaussianUpperBound1}).
\begin{IEEEeqnarray}{rl}
\frac{1}{2}\ln P_0 + c\leq
&\frac{\hat t_{1}}{2}\ln\left(\frac{h_{01}^{2}P_0}{\hat t_{1}}\right)+
\frac{\hat t_{2}}{2}\ln\left(\frac{h_{02}^{2}P_0}{\hat t_{2}}\right)+
\frac{\hat t_{3}}{2}\ln\left(\frac{(h_{01}^2 + h_{02}^2)P_0}{\hat t_{3}}\right)+\nonumber\\
&\frac{\hat t_{1}^2}{2h_{01}^2P_0}+\frac{\hat t_{2}^2}{2h_{02}^2P_0}+\frac{\hat t_{3}^2}{2(h_{01}^2+h_{02}^2)P_0}\nonumber\\
&=\frac{\left(1 - \hat t_{4}\right)}{2}\ln P_0 + \frac{\hat t_{1}}{2}\ln h_{01}^2 + \frac{\hat t_{2}}{2}\ln h_{02}^2 + \frac{\hat t_{3}}{2}\ln\left(h_{01}^2 + h_{02}^2\right)\nonumber\\
&- \frac{\hat t_{1}}{2}\ln \hat t_{1} - \frac{\hat t_{2}}{2}\ln \hat t_{2} - \frac{\hat t_{3}}{2}\ln \hat t_{3} + \frac{\hat t_{1}^2}{2h_{01}^2P_0} +
\frac{\hat t_{2}^2}{2h_{02}^2P_0} + \frac{\hat t_{3}^2}{2\left(h_{01}^2 + h_{02}^2\right)P_0}.\label{eq:UP3}
\end{IEEEeqnarray}
Note that in deriving (\ref{eq:T6_R_DPC}) and (\ref{eq:UP3}), the
following inequality is applied to lower/upper-bound the
corresponding terms:
\begin{equation}\label{eq:ineq1}
\ln(x)\leq\ln(1+x)\leq\ln(x)+\frac{1}{x}, \forall x>0.
\end{equation}
Consequently, we have
\begin{eqnarray}
\hat t_{4} &\leq& \frac{1}{\ln P_0}\left( 2c + \hat t_{1}\ln h_{01}^2
+ \hat t_{2}\ln h_{02}^2 + \hat t_{3}\ln\left(h_{01}^2 + h_{02}^2\right)
- \hat t_{1}\ln \hat t_{1} - \hat t_{2}\ln \hat t_{2} - \hat t_{3}\ln \hat t_{3}\right)\nonumber\\
&+& \frac{1}{\ln P_0}\left(\frac{\hat t_{1}^2}{h_{01}^2P_0}
+ \frac{\hat t_{2}^2}{h_{02}^2P_0} + \frac{\hat t_{3}^2}{\left(h_{01}^2
+ h_{02}^2\right)P_0}\right).\nonumber
\end{eqnarray}

Hence, we can bound the optimum value of $\hat t_4$ in (\ref{eq:GaussianUpperBound1}) as
\begin{eqnarray}\label{eq:t4Bound}
0 \leq \hat t_{4} \leq O\left(\frac{1}{\log P_0}\right).
\end{eqnarray}
Similarly, by considering the fourth cut
in~(\ref{eq:GaussianUpperBound1}), we can derive another bound on
the optimum value of $\hat t_{3}$ as follows:
\begin{eqnarray}\label{eq:t3Bound}
0 \leq \hat t_{3} \leq O\left(\frac{1}{\log P_0}\right).
\end{eqnarray}
Applying the inequality between (\ref{eq:T6_R_DPC}) and the term corresponding
to the second cut in (\ref{eq:GaussianUpperBound1}), knowing (from (\ref{eq:t4Bound}) and (\ref{eq:t3Bound})) the
fact that $\hat t_{3}\leq\frac{c_{3}}{\ln P_0}$, and
$\hat t_{4}\leq\frac{c_{4}}{\ln P_0}$ (where $c_{3}$ and $c_{4}$ are
constants), and using inequalities (\ref{eq:ineq1}), and
\begin{eqnarray}\label{eq:ineq2}
\ln(1 + x) \leq x, \forall x\geq0,
\end{eqnarray}
we obtain
\begin{IEEEeqnarray}{rl}
&\frac{1}{2}\ln P_0 + c \leq\nonumber\\
&\frac{\hat t_{2}}{2} \ln \left(\frac{h_{02}^2h_{13}^2\gamma_{1}P_0^2}{\hat t_{2}^2}
\left(1 + \frac{\hat t_{2}}{\gamma_{1}h_{13}^2P_0} + \frac{\hat t_{2}\left(h_{12}^2 + h_{13}^2\right)}{h_{02}^2h_{13}^2P_0} + \frac{\hat t_{2}h_{12}}{h_{13}^2h_{02}\sqrt{\gamma_{1}}P_0}\right)\right) + \nonumber\\
&\frac{\hat t_{3}}{2} \ln\left(\frac{h_{02}^2P_0}{\hat t_{3}}\right) + \frac{\hat t_{4}}{2} \ln\left(\frac{h_{13}^2\gamma_{1}P_0}{\hat t_{4}}\right)+ \nonumber\\
&\frac{\hat t_{2}^3}{2\left(\hat t_{2}h_{02}^2P_0 + \hat t_{2}\gamma_{1}\left(h_{12}^2 + h_{13}^2\right)P_0 + 2\hat t_{2}h_{02}h_{12}\sqrt{\gamma_{1}}P_0 + h_{02}^2h_{13}^2\gamma_{1}P_0^2\right)}+\nonumber\\
&\frac{\hat t_{3}^2}{2h_{02}^2P_0} + \frac{\hat t_{4}^2}{2\gamma_{1}h_{13}^2P_0}\nonumber\\
&\leq \hat t_{2} \ln P_0 + \frac{\hat t_{2}}{2}\ln\left(\frac{h_{02}^2h_{13}^2\gamma_{1}}{\hat t_{2}^2}\right)
+ \frac{\hat t_{2}^2}{2\gamma_{1}h_{13}^2P_0} + \frac{\hat t_{2}^2\left(h_{12}^2 + h_{13}^2\right)}{2h_{02}^2h_{13}^2P_0} + \frac{\hat t_{2}^2h_{12}}{2h_{13}^2h_{02}\sqrt{\gamma_{1}}P_0} + \nonumber\\
&\frac{c_{3}}{2\ln P_0}\ln h_{02}^2 - \frac{c_{3}}{2 \ln P_0}\ln \hat t_{3} + \frac{c_{3}}{2} + \frac{c_{4}}{2\ln P_0}\ln \gamma_{1}h_{13}^2 - \frac{c_{4}}{2 \ln P_0}\ln \hat t_{4} + \frac{c_{4}}{2}+\nonumber\\
&\frac{\hat t_{2}^3}{2\left(\hat t_{2}h_{02}^2P_0 + \hat t_{2}\gamma_{1}\left(h_{12}^2 + h_{13}^2\right)P_0 + 2\hat t_{2}h_{02}h_{12}\sqrt{\gamma_{1}}P_0 + h_{02}^2h_{13}^2\gamma_{1}P_0^2\right)}+\nonumber\\
&\frac{\hat t_{3}^2}{2h_{02}^2P_0} + \frac{\hat t_{4}^2}{2\gamma_{1}h_{13}^2P_0}\nonumber
\end{IEEEeqnarray}
Therefore, we have
\begin{IEEEeqnarray}{rl}
\frac{1}{2}\ln P_0 + c &\leq \hat t_{2}\ln P_0 + \acute{c} \nonumber\\ &+O\left(\frac{1}{\ln P_0}\right) + O\left(\frac{1}{P_0}\right).\nonumber
\end{IEEEeqnarray}
Hence,
\begin{eqnarray}\label{eq:t2}
\frac{1}{2} - \frac{c_{2}}{\log P_0}\leq \hat t_{2}.
\end{eqnarray}
Similarly, from the third cut of~(\ref{eq:GaussianUpperBound1}), for $\hat t_{1}$ we have
\begin{eqnarray}\label{eq:t1}
\frac{1}{2} - \frac{c_{1}}{\log P_0}\leq \hat t_{1}.
\end{eqnarray}
From~(\ref{eq:t2}) and~(\ref{eq:t1}), and also the fact that
$\hat t_{1}+\hat t_{2}+\hat t_{3}+\hat t_{4} = 1$, we obtain
\begin{IEEEeqnarray}{rl}
&\frac{1}{2} - \frac{c_{2}}{\log P_0}\leq \hat t_{2}\leq\frac{1}{2} + \frac{c_{1}}{\log P_0},\label{eq:t2Bound}\\
&\frac{1}{2} - \frac{c_{1}}{\log P_0}\leq \hat t_{1}\leq\frac{1}{2} + \frac{c_{2}}{\log P_0}.\label{eq:t1Bound}
\end{IEEEeqnarray}
Hence, from (\ref{eq:t4Bound}), (\ref{eq:t3Bound}), (\ref{eq:t2Bound}), and (\ref{eq:t1Bound}) as $P_0\to \infty$, $\hat t_{3}$, $\hat t_{4}\to 0$ and $\hat t_{1}$,
$\hat t_{2} \to 0.5$. This proves the first part of the Theorem.

Moreover, knowing that each term corresponding to the four cuts in (\ref{eq:GaussianUpperBound1}) is greater than $0.5 \ln (P_0) + c$ and as  $\hat t_1,\hat t_2$ are strictly above zero (approaching $0.5$), we can easily conclude that
\begin{equation}
\hat P_{0}^{(1)},\hat P_{0}^{(2)}, \hat P_{1}^{(2)}, \hat P_{2}^{(1)}  \sim \Theta \left( P_0 \right).
\end{equation}

Now, we prove that the DPC scheme with the parameters
$t_1= \hat t_1 + \frac{\hat t_3 + \hat t_4}{2}$, $t_2= \hat t_2 + \frac{\hat t_3
+ \hat t_4}{2}$, $P_{0}^{(1)}= \hat P_{0}^{(1)}$ and $P_{0}^{(2)}= \hat P_{0}^{(2)}$, where
$\hat t_1,\cdots, \hat t_4, \hat P_{0}^{(1)}, \hat P_{0}^{(2)}$ are the parameters
corresponding to the maximum value of
(\ref{eq:GaussianUpperBound1}), achieves the capacity with a gap no
more than $O \left( \frac{1}{\log P_0} \right)$. To prove
this, we show that each of the four terms in (\ref{eq:R_DPC}) is no
more than $O \left( \frac{1}{\log P_0} \right)$ below the
corresponding term (from the same cut) in
(\ref{eq:GaussianUpperBound1}). To show this, for the first cut we
have
\begin{IEEEeqnarray}{rl}
&\hat t_{1}C\left(\frac{h_{01}^{2}\hat P_{0}^{(1)}}{\hat t_{1}}\right)+
\hat t_{2}C\left(\frac{h_{02}^{2}\hat P_{0}^{(2)}}{\hat t_{2}}\right)+
\hat t_{3}C\left(\frac{(h_{01}^2 + h_{02}^2)\hat P_{0}^{(3)}}{\hat t_{3}}\right) - t_{1}C\left(\frac{h_{01}^2 {P_{0}^{(1)}}}{t_{1}}\right) -
t_{2}C\left(\frac{h_{02}^2{P_{0}^{(2)}}}{t_{2}}\right) \stackrel{(a)}{\leq } \nonumber \\
&\frac{\hat t_1}{2} \ln \left( \frac{h_{01}^{2}\hat P_{0}^{(1)}}{\hat t_{1}} \right) + \frac{\hat t_2}{2} \ln \left(\frac{h_{02}^{2}\hat P_{0}^{(2)}}{\hat t_{2}}\right) + \hat t_3 C \left( \frac{(h_{01}^2 + h_{02}^2)\hat P^{(3)}_0}{\hat t_{3}} \right) - \left(\frac{\hat t_1}{2} + \frac{\hat t_3 + \hat t_4}{4}\right) \ln \left(\frac{h_{01}^2\hat P_{0}^{(1)}}{t_{1}}\right) \nonumber \\&  - \left(\frac{\hat t_2}{2} + \frac{\hat t_3 + \hat t_4}{4}\right) \ln \left(\frac{h_{02}^2\hat P_{0}^{(2)}}{t_{2}}\right)  + \frac{\hat t_1^2}{2h_{01}^{2}\hat P_{0}^{(1)}} + \frac{\hat t_2^2}{2h_{02}^{2}\hat P_{0}^{(2)}}  \stackrel{(b)}{\lesssim }
\nonumber \\
&\frac{\hat t_1}{2} \ln \left( \frac{h_{01}^{2}\hat P_{0}^{(1)}}{\hat t_{1}} \right) + \frac{\hat t_2}{2} \ln \left(\frac{h_{02}^{2}\hat P_{0}^{(2)}}{\hat t_{2}}\right) + \frac{\hat t_3}{2} \ln \left( \frac{(h_{01}^2 + h_{02}^2) P_0}{\hat t_{3}+\hat t_{1}} \right) - \left(\frac{\hat t_1}{2} + \frac{\hat t_3 + \hat t_4}{4}\right) \ln \left(\frac{h_{01}^2\hat P_{0}^{(1)}}{t_{1}}\right) \nonumber \\&  - \left(\frac{\hat t_2}{2} + \frac{\hat t_3 + \hat t_4}{4}\right) \ln \left(\frac{h_{02}^2\hat P_{0}^{(2)}}{t_{2}}\right)  + O\left( \frac{1}{\log P_0} \right) \stackrel{(c)}{\lesssim }
\nonumber \\
& \frac{\hat t_3}{2} \ln \left( \frac{P_0}{\sqrt {\hat P_{0}^{(1)} \hat P_{0}^{(2)}}} \right) - \frac{\hat t_4}{4} \ln \left( \hat P_{0}^{(1)} \hat P_{0}^{(2)} \right) + O \left( \frac{1}{\log P_0} \right) \stackrel{(d)}{\lesssim}  O \left( \frac{1}{\log P_0} \right) . \label{eq:T6_p1}
\end{IEEEeqnarray}
Here, $(a)$ follows from (\ref{eq:ineq1}), noting the function $\hat t_1 \ln (P_0 - x - y) + \hat t_2 \ln (y) + \hat t_3 \ln \left(\hat t_3 + \left( h_{01}^2 + h_{02}^2 \right)x \right)$ takes its maximum value at $x \leq \frac{\hat t_3}{ \hat t_3 + \hat t_1}P_0$ and hence substituting $\hat P_0^{(3)}=\frac{\hat t_3}{ \hat t_3 + \hat t_1}P_0$ and finally noting $\hat P_0^{(1)}, \hat P_0^{(2)} \sim \Theta (P_0)$ result in $(b)$ , $(c)$ follows from
$\hat t_3, \hat t_4 \sim O \left( \frac{1}{\log P_0} \right) $ and
$\ln \left( \frac{t_1}{\hat t_1} \right) \sim O \left(
\frac{1}{\log P_0} \right)$, and finally $(d)$ follows from $
\hat P_{0}^{(1)}, \hat P_{0}^{(2)} \sim \Theta ( P_0 )$.

Next, we bound the difference between the terms in the fourth cut of
(\ref{eq:GaussianUpperBound1}) and the fourth term in $C_{DPC}^{low}$ 
\begin{IEEEeqnarray}{rl}
& \hat t_{1}C\left(\frac{h_{23}^2\hat P_{2}^{(1)}}{\hat t_{1}}\right)+
\hat t_{2}C\left(\frac{h_{13}^2\hat P_{1}^{(2)}}{\hat t_{2}}\right)+
\hat t_{4}C\left(\frac{h_{13}^2\hat P_{1}^{(4)}+h_{23}^2\hat P_{2}^{(4)}+2 h_{13}h_{23}\sqrt{\hat P_{1}^{(4)}\hat P_{2}^{(4)}}}{\hat t_{4}}\right)
\nonumber \\ & -t_{1}C\left(\frac{h_{23}^2P_2}{t_{1}}\right) -
t_{2}C\left(\frac{h_{13}^2P_1}{t_{2}}\right) \stackrel{(a)}{\lesssim}  \nonumber \\
& \frac{\hat t_{1}}{2} \ln \left(\frac{h_{23}^2\hat P_{2}^{(1)}}{\hat t_{1}}\right)+
\frac{\hat t_{2}}{2} \ln \left(\frac{h_{13}^2\hat P_{1}^{(2)}}{\hat t_{2}}\right)+
\hat t_{4}C\left(\frac{h_{13}^2\hat P_{1}^{(4)}+h_{23}^2\hat P_{2}^{(4)}+2 h_{13}h_{23}\sqrt{\hat P_{1}^{(4)}\hat P_{2}^{(4)}}}{\hat t_{4}}\right)
 \nonumber \\ & - \left(\frac{\hat t_1}{2} + \frac{\hat t_3 + \hat t_4}{4} \right) \ln \left(\frac{h_{23}^2P_2}{t_{1}}\right) -
\left(\frac{\hat t_2}{2} + \frac{\hat t_3 + \hat t_4}{4} \right) \ln \left(\frac{h_{13}^2P_1}{t_{2}}\right) + O \left( \frac{1}{P_0} \right) \stackrel{(b)}{\lesssim} \nonumber \\
& \frac{\hat t_{1}}{2} \ln \left(\frac{h_{23}^2 P_2}{\hat t_{1}}\right)+
\frac{\hat t_{2}}{2} \ln \left(\frac{h_{13}^2 P_1}{\hat t_{2}}\right)+
\hat t_{4} \ln \left(h_{13} \sqrt {\frac {P_1}{\hat t_2 + \hat t_4}}+h_{23} \sqrt {\frac {P_2}{\hat t_1 + \hat t_4}}\right)
 \nonumber \\ & - \left(\frac{\hat t_1}{2} + \frac{\hat t_3 + \hat t_4}{4} \right) \ln \left(\frac{h_{23}^2P_2}{t_{1}}\right) -
\left(\frac{\hat t_2}{2} + \frac{\hat t_3 + \hat t_4}{4} \right) \ln \left(\frac{h_{13}^2P_1}{t_{2}}\right) + O \left( \frac{1}{P_0} \right) \stackrel{(c)}{\lesssim} \nonumber \\
& \frac{\hat t_4}{2} \ln  \left( \frac{2}{\sqrt {(\hat t_1 + \hat t_4)(\hat t_2 + \hat t_4)} } + \frac{h_{13}}{h_{23}(\hat t_2 + \hat t_4)} \sqrt{\frac{P_1}{P_2}}+ \frac{h_{23}}{(\hat t_1 + \hat t_4)h_{13}} \sqrt{\frac{P_2}{P_1}} \right) - \frac{\hat t_3}{4} \ln \left( P_1 P_2 \right) + O \left( \frac{1}{\log P_0} \right) \stackrel{(d)}{\lesssim} \nonumber \\
& O \left( \frac{1}{\log P_0} \right). \label{eq:T6_p2}
\end{IEEEeqnarray}
Here, $(a)$ follows from (\ref{eq:ineq1}) and noting $\hat P_{1}^{(2)} , \hat P_{2}^{(1)} \sim \Theta (P_0)$, noting the function $\hat t_1 \ln (P_2 - y) + \hat t_2 \ln (P_1 - x) + \hat t_4 \ln \left( \hat t_4 + \left( h_{13} \sqrt x + h_{23} \sqrt y \right)^2 \right)$ takes its maximum value at $x \leq \frac{\hat t_4}{\hat t_4 + \hat t_2} P_1, y \leq \frac{\hat t_4}{\hat t_4 + \hat t_1} P_2$ and hence substituting $\hat P_1^{(4)} = \frac{\hat t_4}{\hat t_4 + \hat t_2} P_1$ and $\hat P_2^{(4)} = \frac{\hat t_4}{\hat t_4 + \hat t_1} P_2$ result in $(b)$, $(c)$ follows from $\hat t_3, \hat t_4 \sim  O \left(
\frac{1}{\log P_0} \right)$ and $\hat t_1, \hat t_2 \sim  0.5 + O
\left( \frac{1}{\log P_0} \right)$, and finally $(d)$ follows from the facts that $\frac{P_1}{P_2} \sim \Theta (1)$, $\hat t_1 + \hat t_4, \hat t_2 + \hat t_4 \sim \Theta (1)$, and $\hat t_4 \sim O (\frac{1}{ \log P_0})$.

Next, we bound the difference between the terms in the second cut of
(\ref{eq:GaussianUpperBound1}) and the second term in $C_{DPC}^{low}$ 
\begin{IEEEeqnarray}{l}
 \hat t_{2}C\left(\frac{h_{02}^2\hat P_{0}^{(2)}}{\hat t_{2}} + \frac{(h_{12}^2 + h_{13}^2)\hat P_{1}^{(2)}}{\hat t_{2}}
+ \frac{2h_{02}h_{12}\sqrt{\hat P_{0}^{(2)}\hat P_{1}^{(2)}}}{\hat t_{2}}
+ \frac{h_{02}^2h_{13}^2\hat P_{0}^{(2)}\hat P_{1}^{(2)}}{\hat t_{2}^{2}}\right)  + \hat t_{3}C\left(\frac{h_{02}^{2}\hat P_{0}^{(3)}}{\hat t_{3}}\right) \nonumber \\
+\hat t_{4}C\left(\frac{h_{13}^2\hat P_{1}^{(4)}}{\hat t_{4}}\right)-
t_{2}C\left(\frac{h_{02}^2P_{0}^{(2)}}{t_{2}}\right) -
t_{2}C\left(\frac{h_{13}^2P_1}{t_{2}}\right) \stackrel{(a)}{\lesssim} \nonumber \\
 \frac{\hat t_2}{2} \ln \left( \frac{h_{02}^2h_{13}^2\hat P_{0}^{(2)} \hat P_1^{(2)}}{\hat t_{2}^{2}} \right) + \hat t_{3} C \left(\frac{h_{02}^{2}\hat P_0^{(3)}}{\hat t_{3}}\right)
+\hat t_{4} C \left(\frac{h_{13}^2 \hat P_1^{(4)}}{\hat t_{4}}\right) - \left( \frac{\hat t_2}{2} + \frac{\hat t_3 + \hat t_4}{4} \right) \ln \left(\frac{h_{02}^2 h_{13}^2  \hat P_{0}^{(2)} P_1}{t_2^2} \right)  + O
\left( \frac{1}{P_0} \right) \stackrel{(b)}{\lesssim} \nonumber \\
 \frac{\hat t_2}{2} \ln \left( \frac{h_{02}^2h_{13}^2\hat P_{0}^{(2)} P_1}{\hat t_{2}^{2}} \right) + \frac{\hat t_{3}}{2} \ln \left(\frac{h_{02}^{2} P_0}{\hat t_{3} + \hat t_{2}}\right)
+\frac{\hat t_{4}}{2} \ln \left(\frac{h_{13}^2 P_1}{\hat t_{4} + \hat t_{2}}\right) \nonumber \\
 - \left( \frac{\hat t_2}{2} + \frac{\hat t_3 + \hat t_4}{4} \right) \ln \left(\frac{h_{02}^2  \hat P_{0}^{(2)}}{t_{2}}\right)  - \left( \frac{\hat t_2}{2} + \frac{\hat t_3 + \hat t_4}{4} \right)
\ln \left(\frac{h_{13}^2P_1}{t_{2}}\right) + O
\left( \frac{1}{P_0} \right) \stackrel{(c)}{\lesssim} \nonumber \\
 \frac{\hat t_3}{4} \ln \left( \frac{P_0^2}{\hat P_{0}^{(2)} P_1} \right) + \frac{\hat t_4}{4} \ln \left( \frac{P_1}{\hat P_{0}^{(2)} } \right) +  O
\left( \frac{1}{ \log P_0} \right) \stackrel{(d)}{\lesssim} O \left( \frac{1}{ \log P_0} \right).\label{eq:T6_p3}
\end{IEEEeqnarray}
Here, $(a)$ follows from (\ref{eq:ineq1}), the fact that $P_{0}^{(2)} =
\hat P_{0}^{(2)} \sim \Theta \left( P_0 \right)$ and upper-bounding
$\hat P_{0}^{(3)} \leq P_0$, $\hat P_{1}^{(4)} \leq P_1$, noting the facts that $\hat P_0^{(2)} + \hat P_0^{(3)} \leq P_0$ and $\hat P_1^{(2)} + \hat P_1^{(4)} = P_1$, the functions $\hat t_2 \ln (P_0 - x)+ \hat t_3 \ln \left(\hat t_3 + h_{02}^2 x \right)$ and $\hat t_2 \ln (P_1 - y)+ \hat t_4 \ln \left(\hat t_4 + h_{13}^2 y \right)$ are maximized at $x \leq \frac{\hat t_3}{\hat t_2 + \hat t_3} P_0$ and $y \leq \frac{\hat t_4}{\hat t_2 + \hat t_4} P_1$, hence, substituting $\hat P_0^{(3)}=\frac{\hat t_3}{\hat t_2 + \hat t_3} P_0$ and $\hat P_1^{(4)}=\frac{\hat t_4}{\hat t_2 + \hat t_4} P_1$ upper-bounds the expression which results in $(b)$, $(c)$ follows from $\hat t_3,\hat t_4 \sim
O \left( \frac{1}{\log P_0} \right), \hat t_1, \hat t_2 \sim 0.5 +
O \left( \frac{1}{\log P_0} \right)$, and finally $(d)$
follows from the fact that $\hat P_{0}^{(2)}, P_1 \sim \Theta
\left( P_0 \right)$ and also $\hat t_3,\hat t_4 \sim O \left(
\frac{1}{\log P_0} \right)$.

Noting that the second and the third cuts are the same, and using
the same argument as in (\ref{eq:T6_p3}), we can bound the
difference between the terms in the third cut of
(\ref{eq:GaussianUpperBound1}) and the third term in $C_{DPC}^{low}$ as
\begin{IEEEeqnarray}{rl}
&\hat t_{1}C\left(\frac{h_{01}^2\hat P_{0}^{(1)}}{\hat t_{1}} + \frac{(h_{12}^2 + h_{23}^2)\hat P_{2}^{(1)}}{\hat t_{1}}
+ \frac{2h_{01}h_{12}\sqrt{\hat P_{0}^{(1)}\hat P_{2}^{(1)}}}{\hat t_{1}} + \frac{h_{01}^2h_{23}^2\hat P_{0}^{(1)}\hat P_{2}^{(1)}}{\hat t_{1}^{2}}\right) \nonumber \\
&+
\hat t_{3}C\left(\frac{h_{01}^{2}\hat P_{0}^{(3)}}{\hat t_{3}}\right)
+\hat t_{4}C\left(\frac{h_{23}^2\hat P_{2}^{(4)}}{\hat t_{4}}\right)  - t_{1}C\left(\frac{h_{01}^2P_{0}^{(1)}}{t_{1}}\right) -
 t_{1}C\left(\frac{h_{23}^2P_2}{t_{1}}\right) \leq O \left( \frac{1}{\log P_0} \right). \label{eq:T6_p4}
\end{IEEEeqnarray}
Observing (\ref{eq:T6_p1}), (\ref{eq:T6_p2}), (\ref{eq:T6_p3}) and
(\ref{eq:T6_p4}), completes the proof of the Theorem.
\end{proof}
\begin{Theorem} \label{thm:Th5}
In low SNR scenarios, assuming $P_1=\gamma_1 P_0,~P_2=\gamma_2 P_0$ with $\gamma_1, \gamma_2$ constants independent of the SNR, when the power available for the source and each relay tends to zero and $\left(h_{13}\sqrt{\gamma_1}+h_{23}\sqrt{\gamma_2}\right)^2\leq\min\left(h_{01}^2,h_{02}^2\right)$, the ratio of the achievable rate of the simultaneous relaying protocol based on DDF to cut-set upper bound goes to 1. In this scenario $t_3 = t_4 = \frac{1}{2}$, and no private messages should be transmitted.
\end{Theorem}
\begin{proof}
By the same argument as in
Theorem~\ref{thm:Th4} and considering only the fourth cut, we obtain another upper bound on the capacity.
By the following inequality
\begin{eqnarray}
\ln(1+x)\leq x.
\end{eqnarray}
we can bound the upper bound on the capacity as
\begin{eqnarray}\label{eq:lowSNR0}
C^{up}&\leq&\frac{\left(h_{13}\sqrt{\gamma_1}+h_{23}\sqrt{\gamma_2}\right)^2P_0}{2\ln2}.
\end{eqnarray}

Now, assuming $t_1=t_2=0, t_3=t_4=\frac{1}{2}$, and transmitting just the common message, we can achieve the following rate $C_{DDF}^{low}$:
\begin{eqnarray}
C_{DDF}^{low}=\min\left(\frac{1}{2}C\left(2h_{02}^2P_0\right),
\frac{1}{2}C\left(2\left(h_{13}\sqrt{\gamma_1}+h_{23}\sqrt{\gamma_2}\right)^2P_0\right)\right).
\end{eqnarray}
According to the Taylor expansion of $\ln (1+x)$ at $x=0$, we have
\begin{eqnarray}
x - \frac{x^2}{2}\leq\ln\left(1+x\right),
\end{eqnarray}
Hence,
\begin{eqnarray}\label{eq:lowSNR2}
\frac{1}{\ln2}\min\left(\frac{h_{02}^2P_0}{2}-\frac{h_{02}^4P_0^2}{2},
\frac{\left(h_{13}\sqrt{\gamma_1}+h_{23}\sqrt{\gamma_2}\right)^2P_0}{2}-
\frac{\left(h_{13}\sqrt{\gamma_1}+h_{23}\sqrt{\gamma_2}\right)^4P_0^2}{2}\right)\leq C_{DDF}^{low}.
\end{eqnarray}
By~(\ref{eq:lowSNR0}),~(\ref{eq:lowSNR2}), and $\left(h_{13}\sqrt{\gamma_1}+h_{23}\sqrt{\gamma_2}\right)^2\leq\min\left(h_{01}^2,h_{02}^2\right)$, we have
\begin{eqnarray}
\lim_{P_0\rightarrow0}\frac{C_{DDF}^{low}}{C^{up}}\rightarrow1.
\end{eqnarray}
\end{proof}

\section{Simulation Result}
In this section, the achievable rate of different proposed schemes, i.e., SSRD, DPC, BME, and BME-DPC are compared with each other and with the upper bound in different channel conditions.

Figure \ref{fig:SSRDmm.eps} compares the achievable rate of the SSRD scheme with that of the DPC scheme for successive relaying and the DDF scheme for simultaneous relaying protocols. Here the symmetric scenario in which $P_1=P_2$ and $h_{01} = h_{02}= h_{12} = h_{13} = h_{23} = 1$ is considered. The upper bound is also included in the figure.

In order to satisfy the condition in Theorem  \ref{thm:Th5}, i.e., $\left(h_{13}\sqrt{\gamma_1}+h_{23}\sqrt{\gamma_2}\right)^2\leq\min\left(h_{01}^2,h_{02}^2\right)$, in Figs.~\ref{fig:SSRDmm.eps}a and b, we also assume $P_0 = P_1 + 10(dB) = P_2 + 10(dB)$ and $P_0 = P_1 + 5(dB) = P_2 + 5(dB)$, respectively.  As the Figs.~\ref{fig:SSRDmm.eps}a and b show, SSRD achievable rate almost coincides with the upper bound over all ranges of SNR. As proved in the previous section, in high SNR scenario, SSRD scheme coincides with DPC and the successive relaying protocol becomes optimum, while in low SNR scenario it coincides with DDF and the simultaneous relaying protocol is optimum.

On the other hand, in Figs.~\ref{fig:SSRDmm.eps}c and d we assume that $P_0 = P_1 = P_2$ and $P_0 = P_1 - 5(dB) = P_2 - 5(dB)$. In this situation, the condition in Theorem~\ref{thm:Th5} is no longer satisfied. Therefore, as these figures show, the ratio of the achievable rate of the SSRD scheme to the cut-set bound, i.e.,  $\frac{C_{SSRD}^{low}}{C^{up}}$ does not tend to one. Furthermore, the achievable rates of the SSRD, DPC, and DDF schemes coincide with each other.
\begin{figure}[tp]
\centering
\includegraphics[scale = .75]{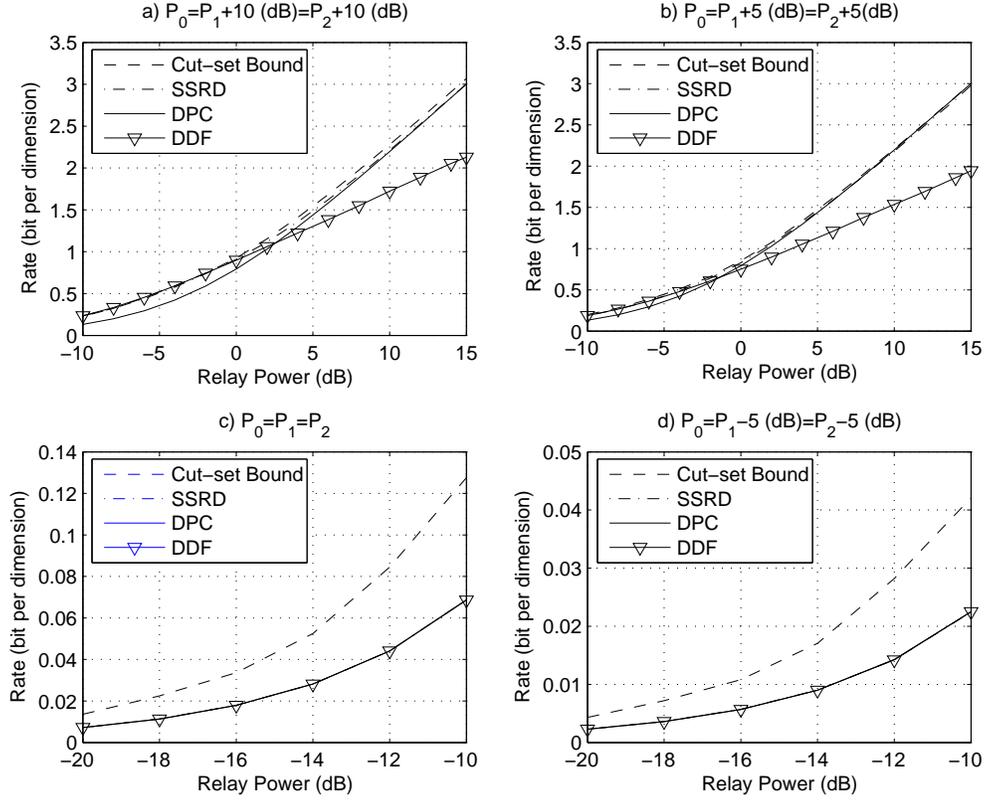}
\caption{\small{Rate versus relay power.}}\label{fig:SSRDmm.eps}
\end{figure}

Figure~\ref{fig:UpDPCBME.eps} compares the achievable
rate of different successive schemes with each other and the
successive cut-set bound. It shows as the inter relay channel becomes
stronger, BME scheme can achieve the successive
cut-set bound, while the achievable rate of the DPC is
independent of that channel. Furthermore, this figure indicates BME-DPC gives a better achievable rate
with respect to BME with successive decoding which
was proposed in~\cite{Yong}.
\begin{figure}[tp]
\centering
\includegraphics[scale = .75]{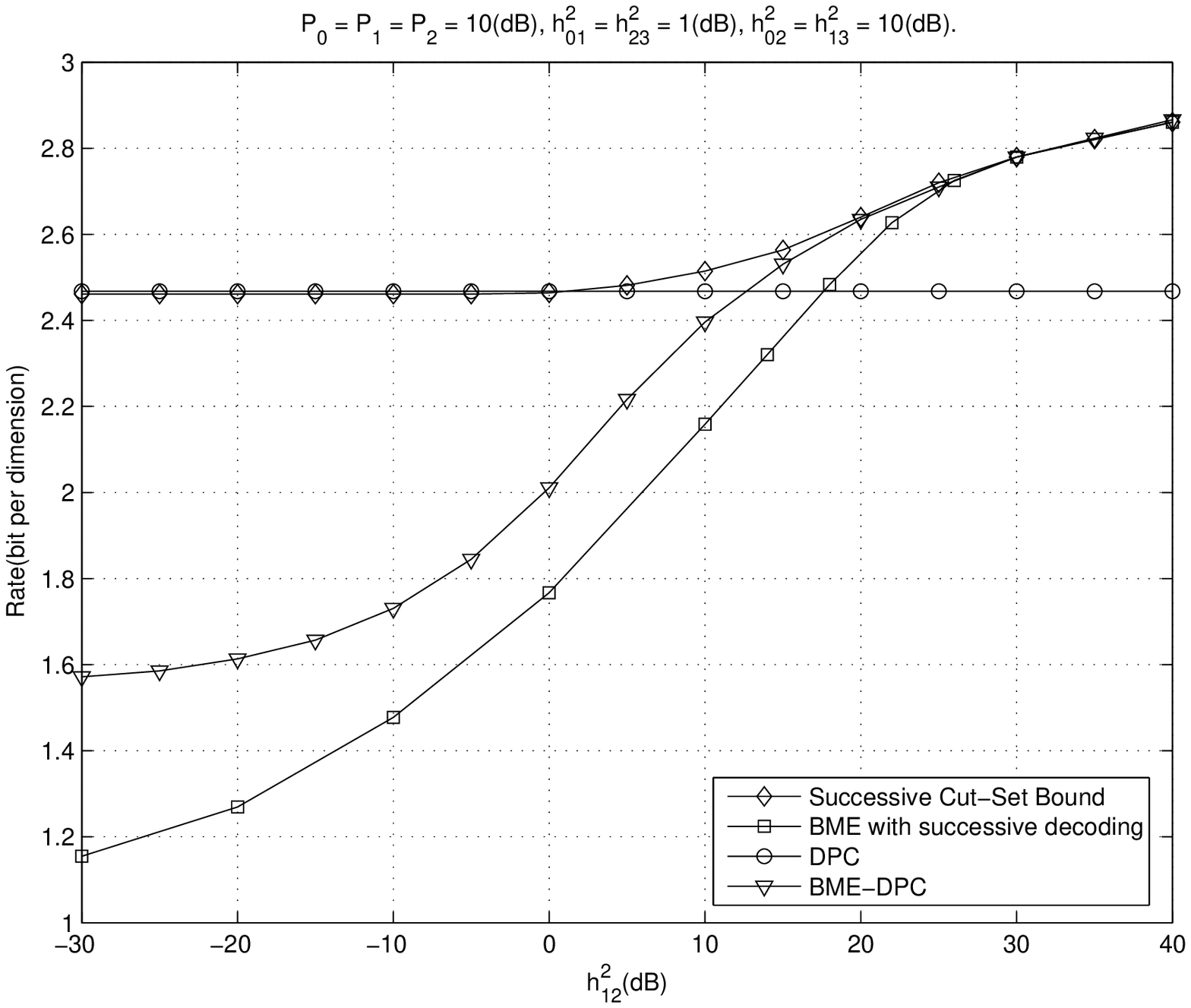}
\caption{\small{Rate versus inter relay gain.}}\label{fig:UpDPCBME.eps}
\end{figure}
\section{Conclusion}
In this paper, we investigated the problem of cooperative strategies for a half-duplex parallel relay channel with two relays. We derived the optimum relay ordering and hence the asymptotic capacity of the half-duplex Gaussian parallel relay channel in low and high SNR scenarios.

\emph{Simultaneous} and \emph{Successive} relaying protocols, associated with two possible relay orderings were proposed. For simultaneous
relaying, each relay employs \emph{DDF}. On the other hand, for successive relaying, we proposed a \emph{Non-Cooperative Coding} scheme based on DPC
and a \emph{Cooperative Coding} scheme based on BME. Moreover, a coding scheme based on the combination of DPC and BME, in which one of the relays uses DPC while the other one employs BME was proposed. We showed that this composite scheme achieves a better rate with respect to cooperative coding based on BME with backward or successive decoding in the Gaussian case.

We also proposed the SSRD scheme as a combination of the simultaneous and successive protocols based on DPC. In high SNR scenarios, we proved that our \emph{Non-Cooperative Coding} scheme based on \emph{DPC} asymptotically achieves the capacity. Hence, in the high SNR scenario, the optimum relay ordering is \emph{Successive}. On the other hand, in low SNR where $\left(h_{13}\gamma_1+h_{23}\gamma_2\right)^2\leq\min\left(h_{01}^2,h_{02}^2\right)$, \emph{DDF} achieves the capacity. Hence, in low SNR scenario and under the condition specified above for the channel coefficients, the optimum relay ordering is \emph{Simultaneous}.

\section*{Appendix A}\label{ap:a}
\subsection*{Proof of Theorem 1}
\underline{\emph{Codebook Construction:}}\\

Let us divide time slot number \emph{b}, $b = 1, 2, \cdots, B+1$ into odd and even numbers. At odd and even time slots, source generates $2^{nr_{AUX}^{(1)}}$ and $2^{nr_{AUX}^{(2)}}$ sequences $\textbf{u}_{0}^{(1)}\left(q_1\right)$ and $\textbf{u}_{0}^{(2)}\left(q_2\right)$ according to $\prod_{i=1}^{t_1n}p(u_{0,i}^{(1)})$ and $\prod_{i=1}^{t_{2}n}p(u_{0,i}^{(2)})$, respectively. Then, source throws $\textbf{u}_{0}^{(1)}$ and $\textbf{u}_{0}^{(2)}$ sequences uniformly into $2^{nR^{(1)}}$ and $2^{nR^{(2)}}$ bins, respectively. Let us denote $\mathcal B_1(w^{(b)})$ and $\mathcal B_2(w^{(b)})$ as the set of sequences at the odd or even time slot that belong to the $w^{(b)}$'th bin, respectively (for odd time slots, $w^{(b)} \leq 2^{nR^{(1)}}$, and for the even time slots, $w^{(b)} \leq 2^{nR^{(2)}}$).

Relay 1 and relay 2 generate $2^{nR^{(1)}}$ and $2^{nR^{(2)}}$ i.i.d $\textbf{x}_{1}^{(2)}$ and $\textbf{x}_{2}^{(1)}$ sequences according to probabilities $\prod_{i=1}^{t_{2}n}p\left(x_{1,i}^{(2)}\right)$ and $\prod_{i=1}^{t_{1}n}p\left(x_{2,i}^{(1)}\right)$.
Furthermore, for all $q_1$ and $q_2$, the source generates double indexed codebooks $\textbf{x}_{0}^{(1)}\left(w^{(b)}|w^{(b-1)},q_1\right)$ and $\textbf{x}_{0}^{(2)}\left(w^{(b)}|w^{(b-1)},q_2\right)$ according to $\prod_{i=1}^{t_{1}n}p(x_{0,i}^{(1)}\mid x_{2,i}^{(1)},u_{0,i}^{(1)})$ and $\prod_{i=1}^{t_{2}n}p(x_{0,i}^{(2)}\mid x_{1,i}^{(2)},u_{0,i}^{(2)})$, respectively.\\

\underline{\emph{Encoding:}}

\emph{Encoding at the source:}

At the odd time slot $b$, the source intends to send the message $w^{(b)}$ to the first relay. In order to do that,  since source knows what it has transmitted during the last time slot to the second relay, it chooses a codeword $\textbf{u}_{0}^{(1)}\left(q_1\right)$ such that $\textbf{u}_{0}^{(1)}\left(q_1\right)\in \mathcal B_1(w^{(b)})$ and $\left(\textbf{u}_{0}^{(1)}\left(q_1\right),\textbf{x}_{2}^{(1)}\left(w^{(b-1)}\right)\right)\in A_{\epsilon}^{(n)}$. Such a task can be done almost surely, if $r_{AUX}^{(1)}-R^{(1)}\geq t_{1}I\left(U_{0}^{(1)};X_{2}^{(1)}\right)$ (See \cite{CoverThomas}). Following that it sends $\textbf{x}_{0}^{(1)}(\textbf{u}_{0}^{(1)},\textbf{x}_{2}^{(1)})$.

At the even time slot $b$, the source sends the message $w^{(b)}$ to the second relay in the similar manner. Such a task can be done almost surely if $r_{AUX}^{(2)}-R^{(2)}\geq t_{2}I\left(U_{0}^{(2)};X_{1}^{(2)}\right)$.

\emph{Encoding at relay 1:}

At the even time slot $b$, relay 1 encodes $w^{(b-1)}\in\{1,\cdots,2^{nR^{(1)}}\}$ to $\textbf{x}_{1}^{(2)}\left(w^{(b-1)}\right)$.

\emph{Encoding at relay 2:}

At the odd time slot $b$, relay 2 encodes $w^{(b-1)}\in\{1,\cdots,2^{nR^{(2)}}\}$ to $\textbf{x}_{2}^{(1)}\left(w^{(b-1)}\right)$.\\

\underline{\emph{Decoding:}}

\emph{Decoding at relay 1:}

At the odd time slot $b$, relay 1 declares $\hat{w}^{(b)} = {w}^{(b)}$ iff all the sequences $\textbf{u}_{0}^{(1)}\left(q_1\right)$ which are jointly typical with $\textbf{y}_{1}^{(1)}$ belong to a unique bin $\mathcal B_1 (\hat{w}^{(b)})$. Therefore, in order to make the probability of error zero, from~\cite{CoverThomas}, we have
\begin{eqnarray}
r_{AUX}^{(1)}\leq t_{1}I\left(U_{0}^{(1)};Y_{1}^{(1)}\right).\label{eq:Dec1}
\end{eqnarray}
According to~(\ref{eq:Dec1}) and the encoding condition at source, we have
\begin{eqnarray}
R^{(1)}\leq t_{1}\left(I(U_{0}^{(1)};Y_{1}^{(1)}) - I(U_{0}^{(1)};X_{2}^{(1)})\right).
\end{eqnarray}

\emph{Decoding at relay 2:}

At the even time slot $b$, relay 2 declares $\hat{w}^{(b)} = {w}^{(b)}$ iff all the sequences $\textbf{u}_{0}^{(2)}\left(q_2\right)$ which are jointly typical with $\textbf{y}_{2}^{(2)}$ belong to a unique bin $\mathcal B_2(\hat{w}^{(b)})$. Therefore, in order to make the probability of error zero, from~\cite{CoverThomas}, we have
\begin{eqnarray}
r_{AUX}^{(2)}\leq t_{2}I\left(U_{0}^{(2)};Y_{2}^{(2)}\right).\label{eq:Dec2}
\end{eqnarray}
According to~(\ref{eq:Dec2}) and the encoding condition at source, we have
\begin{eqnarray}
R^{(2)}\leq t_{2}\left(I(U_{0}^{(2)};Y_{2}^{(2)}) - I(U_{0}^{(2)};X_{1}^{(2)})\right).
\end{eqnarray}

\emph{Decoding at the final destination:}

At the odd time slot $b$, destination declares $\hat{w}^{(b-1)} = w^{(b-1)}$ iff $\left(\textbf{x}_{2}^{(1)}\left(\hat{w}^{(b-1)}\right),\textbf{y}_{3}^{(1)}\right)\in A_{\epsilon}^{(n)}$. Hence, in order to make the probability of error zero, from~\cite{CoverThomas}, we have
\begin{eqnarray}
R^{(1)}\leq t_{1}I(X_{2}^{(1)};Y_{3}^{(1)}).
\end{eqnarray}

Similarly, at the even time slot $b$, we have
\begin{eqnarray}
R^{(2)}\leq t_{2}I(X_{1}^{(2)};Y_{3}^{(2)}).\label{eq:Dec3}
\end{eqnarray}
From the encoding at the source and (\ref{eq:Dec1})-(\ref{eq:Dec3}), we obtain (\ref{eq:DPC1})-(\ref{eq:DPC3}).
\section*{Appendix B}\label{ap:b}

\subsection*{Proof of Theorem 2}

\underline{\emph{Codebook Construction:}}\\

Let us divide the time slots \emph{b}, $b = 1, 2, \cdots, B+2$ into odd and even time slots.
The source generates two codebooks $\textbf{x}_{0}^{(1)}\left(w^{(b)}|w^{(b-1)},s_{1}^{(b-2)}\right)$ and $\textbf{x}_{0}^{(2)}\left(w^{(b)}|w^{(b-1)},s_{2}^{(b-2)}\right)$ of size $2^{nR^{(1)}}$ and $2^{nR^{(2)}}$ corresponding to even and odd time slots, respectively. The first codebook is generated according to the probability $p(\textbf{x}_{0}^{(1)},\textbf{x}_{2}^{(1)}, \textbf{u}_{2}^{(1)}) = \prod_{i = 1}^{t_1n}p(u_{2,i}^{(1)})p
(x_{2,i}^{(1)}|u_{2,i}^{(1)})p(x_{0,i}^{(1)}|x_{2,i}^{(1)},u_{2,i}^{(1)})$, and the second codebook is generated according to the
probability $p(\textbf{x}_{0}^{(2)},\textbf{x}_{1}^{(2)},
\textbf{u}_{1}^{(2)}) = \prod_{i = 1}^{t_{2}n}p(u_{1,i}^{(2)})p
(x_{1,i}^{(2)}|u_{1,i}^{(2)})p(x_{0,i}^{(2)}|x_{1,i}^{(2)},u_{1,i}^{(2)})$.

On the other hand, relay 2 generates $2^{nr_{Bin}^{(1)}}$ i.i.d codewords $\textbf{u}_{2}^{(1)}$ and $2^{nR^{(2)}}$ i.i.d codewords $\textbf{x}_{2}^{(1)}$ according to the probabilities $p(\textbf{u}_{2}^{(1)})=\prod_{i=1}^{t_{1}n}p(u_{2,i}^{(1)})$ and $p(\textbf{x}_{2}^{(1)}\mid\textbf{u}_{2}^{(1)})=\prod_{i=1}^{t_{1}n}
p(x_{2,i}^{(1)}\mid u_{2,i}^{(1)})$ at each odd time slot and relay 1 generates $2^{nr_{Bin}^{(2)}}$ i.i.d codewords $\textbf{u}_{1}^{(2)}$ and $2^{nR^{(1)}}$ i.i.d codewords $\textbf{x}_{1}^{(2)}$ according to the probabilities $p(\textbf{u}_{1}^{(2)})=\prod_{i=1}^{t_{2}n}p(u_{1,i}^{(2)})$ and $p(\textbf{x}_{1}^{(2)}\mid\textbf{u}_{1}^{(2)})=\prod_{i=1}^{t_{2}n}
p(x_{1,i}^{(2)}\mid u_{1,i}^{(2)})$ at each even time slot, respectively.\\

\underline{\emph{Encoding:}}

\emph{Encoding at the source:}

At the odd time slot $b$, source encodes $w^{(b)}\in\{1,\cdots, 2^{nR^{(1)}}\}$ to $\textbf{x}_{0}^{(1)}\left(w^{(b)}|w^{(b-1)},s_{1}^{(b-2)}\right)$ and at the even time slot $b$, it encodes $w^{(b)}\in\{1,\cdots, 2^{nR^{(2)}}\}$ to $\textbf{x}_{0}^{(2)}\left(w^{(b)}|w^{(b-1)},s_{2}^{(b-2)}\right)$ and sends them in odd and even time slots, respectively.

\emph{Encoding at relay 1:}

At the even time slot $b$, relay 1 encodes the bin index $s_{2}^{(b-2)}$ of the message $w^{(b-2)}$ it has received from relay 2 in the previous time slot to $\textbf{u}_{1}^{(2)}\left(s_{2}^{(b-2)}\right)$. Following that, it encodes $w^{(b-1)}$ which was received from the source in time slot $b-1$ to $\textbf{x}_{1}^{(2)}\left(w^{(b-1)}|s_{2}^{(b-2)}\right)$ and sends it.

\emph{Encoding at relay 2:}

At the odd time slot $b$, relay 2 encodes the bin index $s_{1}^{(b-2)}$ of the message $w^{(b-2)}$ it has received from relay 1 in the previous time slot to $\textbf{u}_{2}^{(1)}\left(s_{1}^{(b-2)}\right)$. Following that, it encodes $w^{(b-1)}$ which was received from the source in time slot $b-1$ to $\textbf{x}_{2}^{(1)}\left(w^{(b-1)}|s_{1}^{(b-2)}\right)$ and sends it.

\underline{\emph{Decoding:}}

\emph{Decoding at relay 1:}

Knowing $w^{(b-2)}$ and consequently $s_{1}^{(b-2)}$, at time slot $b$, relay 1 declares $(\hat{w}^{(b-1)},\hat{w}^{(b)}) = (w^{(b-1)},w^{(b)})$ iff there exits a unique $(\hat{w}^{(b-1)},\hat{w}^{(b)})$ such that
\begin{eqnarray}
\left(\textbf{x}_{0}^{(1)}\left(\hat{w}^{(b)}|\hat{w}^{(b-1)},s_{1}^{(b-2)}\right),\textbf{x}_{2}^{(1)}\left(\hat{w}^{(b-1)}|s_{1}^{(b-2)}\right),\textbf{u}_{2}^{(1)}(s_{1}^{(b-2)}),
\textbf{y}_{1}^{(1)}\right)\in A_{\epsilon}^{(n)}.\nonumber
\end{eqnarray}
Hence, in order to make probability of error zero, from the Extended MAC capacity region (See \cite{CoverThomas},~\cite{Slepian1},~\cite{Han}, and~\cite{Prelov}), we have
\begin{IEEEeqnarray}{rl}\label{eq:Ach1}
&R^{(1)}\leq t_{1}I\left(X_{0}^{(1)};Y_{1}^{(1)}\mid X_{2}^{(1)},U_{2}^{(1)}\right),\label{eq:relay1}\\
&R^{(1)} + R^{(2)}\leq
t_{1}I(X_{0}^{(1)},X_{2}^{(1)};Y_{1}^{(1)}\mid
U_{2}^{(1)}).\label{eq:relay11}
\end{IEEEeqnarray}

\emph{Decoding at relay 2:}

Knowing $w^{(b-2)}$ and consequently $s_{2}^{(b-2)}$, at time slot $b$, relay 2 declares $(\hat{w}^{(b-1)},\hat{w}^{(b)}) = (w^{(b-1)},w^{(b)})$ iff there exits a unique $(\hat{w}^{(b-1)},\hat{w}^{(b)})$ such that
\begin{eqnarray}
\left(\textbf{x}_{0}^{(2)}\left(\hat{w}^{(b)}|\hat{w}^{(b-1)},s_{2}^{(b-2)}\right),\textbf{x}_{1}^{(2)}\left(\hat{w}^{(b-1)}|s_{2}^{(b-2)}\right),\textbf{u}_{1}^{(2)}(s_{2}^{(b-2)}),
\textbf{y}_{2}^{(2)}\right)\in A_{\epsilon}^{(n)}. \nonumber
\end{eqnarray}
Hence, in order to make the probability of error zero, from Extended MAC capacity region (See \cite{CoverThomas},~\cite{Slepian1},~\cite{Han}, and~\cite{Prelov}), we have
\begin{IEEEeqnarray}{rl}\label{eq:Ach2}
&R^{(2)}\leq t_{2}I(X_{0}^{(2)};Y_{2}^{(2)}\mid
X_{1}^{(2)},U_{1}^{(2)}),\label{eq:relay2}\\
&R^{(1)} + R^{(2)}\leq
t_{2}I(X_{0}^{(2)},X_{1}^{(2)}
;Y_{2}^{(2)}\mid
U_{1}^{(2)}).\label{eq:relay22}
\end{IEEEeqnarray}

\emph{Decoding at the final destination:}

Decoding at the final destination can be done either \emph{Successively} or \emph{Backwardly} as follows.

\emph{1) Successive Decoding:}

At the end of odd time slot $b$, destination first declares the bin index $\hat{s}_{1}^{(b-2)} = s_{1}^{(b-2)}$ of the message $w^{(b-2)}$ iff
there exists a unique $\hat{s}_{1}^{(b-2)}$ such that $\left(\textbf{u}_{2}^{(1)}(\hat{s}_{1}^{(b-2)}),\textbf{y}_{3}^{(1)}\right)\in A_{\epsilon}^{(n)}$. Hence, in order to make the probability of error zero, from~\cite{CoverThomas} we have
\begin{IEEEeqnarray}{rl}\label{eq:Ach21}
&r_{Bin}^{(1)}\leq t_{1}I(U_{2}^{(1)};Y_{3}^{(1)}).
\end{IEEEeqnarray}
Having decoded the bin index $s_{1}^{(b-2)}$ of the message $w^{(b-2)}$, destination can resolve its uncertainty about the message $w^{(b-2)}$ and declares $\hat{w}^{(b-2)} = w^{(b-2)}$ iff there exists a unique $\hat{w}^{(b-2)}$ such that $\left(\textbf{x}_{1}^{(2)}(\hat{w}^{(b-2)}|s_{2}^{(b-3)}),
\textbf{u}_{1}^{(2)}(s_{2}^{(b-3)}),
\textbf{y}_{3}^{(2)}\right)\in A_{\epsilon}^{(n)}$. Hence, in order to make the probability of error zero, from~\cite{CoverThomas} we have
\begin{IEEEeqnarray}{rl}\label{eq:Ach22}
&R^{(1)}-r_{Bin}^{(1)}\leq t_{2}I(X_{1}^{(2)}
;Y_{3}^{(2)}\mid U_{1}^{(2)}).
\end{IEEEeqnarray}
Using the same argument for the even time slot \emph{b}, we have
\begin{IEEEeqnarray}{rl}
&r_{Bin}^{(2)}\leq t_{2}I(U_{1}^{(2)};
Y_{3}^{(2)}),\label{eq:Ach23}\\
&R^{(2)}-r_{Bin}^{(2)}\leq t_{1}I(X_{2}^{(1)}
;Y_{3}^{(1)}\mid U_{2}^{(1)}).\label{eq:Ach24}
\end{IEEEeqnarray}
From (\ref{eq:Ach21}), (\ref{eq:Ach22}), (\ref{eq:Ach23}), and (\ref{eq:Ach24}), $R^{(1)}$ and $R^{(2)}$ are bounded as follows
\begin{IEEEeqnarray}{rl}
&R^{(1)}\leq t_{2}I\left(X_{1}^{(2)}
;Y_{3}^{(2)}\mid U_{1}^{(2)}\right)+t_{1}I\left(U_{2}^{(1)};
Y_{3}^{(1)}\right),\label{eq:Ach31}\\
&R^{(2)}\leq t_{1}I(X_{2}^{(1)}
;Y_{3}^{(1)}\mid U_{2}^{(1)})+t_{2}I(U_{1}^{(2)};
Y_{3}^{(2)}).\label{eq:Ach32}
\end{IEEEeqnarray}
From (\ref{eq:relay1})-(\ref{eq:relay22}), (\ref{eq:Ach31}), and (\ref{eq:Ach32}), the achievable rate of BME scheme based on successive decoding is equal to
\begin{IEEEeqnarray}{rl}
&C_{BM_{succ}}^{low} = R^{(1)} + R^{(2)} \leq\max_{0\leq t_{1},t_{2},t_{1}+t_{2}=1}\min\left( \right.\\
&\left.
\min{\left(t_{1}I\left(X_{0}^{(1)};Y_{1}^{(1)}\mid X_{2}^{(1)},U_{2}^{(1)}\right), t_{2}I\left(X_{1}^{(2)}
;Y_{3}^{(2)}\mid U_{1}^{(2)}\right)+t_{1}I\left(U_{2}^{(1)};
Y_{3}^{(1)}\right)\right)}+\right.\nonumber\\
&\left.
\min{\left(t_{1}I\left(X_{2}^{(1)}
;Y_{3}^{(1)}\mid U_{2}^{(1)}\right)+t_{2}I\left(U_{1}^{(2)};
Y_{3}^{(2)}\right), t_{2}I\left(X_{0}^{(2)};Y_{2}^{(2)}\mid
X_{1}^{(2)},U_{1}^{(2)}\right)\right)},\right.\nonumber\\
&\left.
t_{1}I\left(X_{0}^{(1)},X_{2}^{(1)};Y_{1}^{(1)}\mid
U_{2}^{(1)}\right), t_{2}I\left(X_{0}^{(2)},X_{1}^{(2)}
;Y_{2}^{(2)}\mid
U_{1}^{(2)}\right)\right).\nonumber
\end{IEEEeqnarray}

\emph{2) Backward Decoding:}

Following receiving the sequence corresponding to the $B+2$'th time slot, destination starts decoding the messages in a backward manner, i.e. from $w^{(B)}$ back to $w^{(1)}$. At the end of odd time slot $b$, knowing the value $s_2^{(b-1)}$ from the received signal in time slot $b+1$, destination declares $\left(\hat{w}^{(b-1)},\hat{s}_{1}^{(b-2)}\right)= \left(w^{(b-1)},s_{1}^{(b-2)}\right)$ iff there exists a unique pair $\left(\hat{w}^{(b-1)},\hat{s}_{1}^{(b-2)}\right)$
such that $f_{Bin}^{(2)}\left(\hat{w}^{(b-1)}\right) = s_2^{(b-1)}$ and
$\left(\textbf{x}_{2}^{(1)}\left(\hat{w}^{(b-1)},\hat{s}_{1}^{(b-2)}\right),
\textbf{u}_{2}^{(1)}\left(\hat{s}_{1}^{(b-2)}\right),\textbf{y}_{3}^{(1)}\right)\in A_{\epsilon}^{(n)}$. Similarly, at the end of even time slot $b$, knowing the value $s_1^{(b-1)}$ for the received signal in time slot $b+1$, destination declares $\left(\hat{w}^{(b-1)},\hat{s}_{2}^{(b-2)}\right)= \left(w^{(b-1)},s_{2}^{(b-2)}\right)$ iff there exists a unique pair $\left(\hat{w}^{(b-1)},\hat{s}_{2}^{(b-2)}\right)$ such that $f_{Bin}^{(1)}\left(\hat w^{(b-1)}\right) = s_1^{(b-1)}$ and
$\left(\textbf{x}_{1}^{(2)}\left(\hat{w}^{(b-1)},\hat{s}_{1}^{(b-2)}\right),
\textbf{u}_{1}^{(2)}\left(\hat{s}_{2}^{(b-2)}\right),\textbf{y}_{3}^{(2)}\right)\in A_{\epsilon}^{(n)}$. Hence, in order to make the probability of error zero, from~\cite{CoverThomas} we have
\begin{IEEEeqnarray}{rl}
&r_{Bin}^{(1)}\leq R^{(1)},\label{eq:Back7}\\
&r_{Bin}^{(2)}\leq R^{(2)},\label{eq:Back8}\\
&R^{(2)}-r_{Bin}^{(2)}\leq t_1I\left(X_2^{(1)};Y_3^{(1)}\mid U_2^{(1)}\right),\label{eq:Back9}\\
&R^{(2)}-r_{Bin}^{(2)}+r_{Bin}^{(1)}\leq t_1I\left(X_{2}^{(1)},U_2^{(1)};Y_{3}^{(1)}\right),\label{eq:Back10}\\
&R^{(1)}-r_{Bin}^{(1)}\leq t_2 I\left(X_1^{(2)};Y_3^{(2)}\mid U_1^{(2)}\right),\label{eq:Back11}\\
&R^{(1)}-r_{Bin}^{(1)}+r_{Bin}^{(2)}\leq t_2I\left(X_{1}^{(2)},U_1^{(2)};Y_{3}^{(2)}\right).\label{eq:Back12}
\end{IEEEeqnarray}
Hence, by employing BME and Backward decoding, the following rate is achievable subject to (\ref{eq:relay1})-(\ref{eq:relay22}) and  (\ref{eq:Back7})-(\ref{eq:Back12}) constraints.
\begin{IEEEeqnarray}{rl}
&C_{BME_{back}}^{low}=R^{(1)} + R^{(2)}.\label{eq:R_Back}
\end{IEEEeqnarray}

\underline{\emph {Optimum input distributions}}

Now, we prove there exists input probability distributions ($p(x_0^{(1)},x_2^{(1)},u_2^{(1)})$ and $p(x_0^{(2)},x_1^{(2)},u_1^{(2)})$) which maximize \eqref{eq:R_Back} and have the following property: $u_2^{(1)}$ is independent from $(x_0^{(1)},x_2^{(1)})$ and $u_1^{(2)}$ is independent from $(x_0^{(2)},x_1^{(2)})$. To prove this, consider $p(x_0^{(1)},x_2^{(1)},u_2^{(1)})$ and $p(x_0^{(2)},x_1^{(2)},u_1^{(2)})$ along with $t_1, t_2$ which  maximize \eqref{eq:R_Back} subject to the required constraints. Let us define $\hat p(x_0^{(1)},x_2^{(1)},u_2^{(1)})$ and $\hat p(x_0^{(2)},x_1^{(2)},u_1^{(2)})$ as
\begin{IEEEeqnarray}{rl}
&\hat p(x_0^{(1)},x_2^{(1)},u_2^{(1)}) = p(u_2^{(1)})p(x_0^{(1)}, x_2^{(1)}),\label{eq:p1}\\
&\hat p(x_0^{(2)},x_1^{(2)},u_1^{(2)}) = p(u_1^{(2)})p(x_0^{(2)}, x_1^{(2)}),\label{eq:p2}
\end{IEEEeqnarray}
Now, we show that $\hat p(x_0^{(1)},x_2^{(1)},u_2^{(1)})$ and $\hat p(x_0^{(2)},x_1^{(2)},u_1^{(2)})$ along with $t_1, t_2$ achieve at least the same rate as the optimum one. Let us denote the values of mutual information and entropy with respect to the input distributions $p, \hat p$ by $I_{p}, H_{p}$ and $ I_{\hat p}, H_{\hat p}$, respectively. The right-hand sides of \eqref{eq:Back9}-\eqref{eq:Back12} with respect to $p$ can be upper-bounded by the ones corresponding to $\hat p$ as follows
\begin{IEEEeqnarray}{lcl}
t_1I_p\left(X_2^{(1)};Y_3^{(1)}\mid U_2^{(1)}\right) & \stackrel{(a)}{\leq} & t_1I_p \left(X_2^{(1)};Y_3^{(1)}\right) = t_1I_{\hat p} \left(X_2^{(1)};Y_3^{(1)}\right),\label{eq:khar}\\
t_1I_p\left(X_{2}^{(1)},U_2^{(1)};Y_{3}^{(1)}\right) & \stackrel{(a)}{=} & t_1 I_p\left(X_{2}^{(1)};Y_{3}^{(1)}\right) = t_1 I_{\hat p}\left(X_{2}^{(1)};Y_{3}^{(1)}\right),\\
t_2 I_p\left(X_1^{(2)};Y_3^{(2)}\mid U_1^{(2)}\right) & \stackrel{(b)}{\leq} &
t_2 I_p\left(X_1^{(2)};Y_3^{(2)}\right) = t_2 I_{\hat p}\left(X_1^{(2)};Y_3^{(2)}\right),\\
t_2I_p\left(X_{1}^{(2)},U_1^{(2)};Y_{3}^{(2)}\right)& \stackrel{(b)}{=} &
t_2I_p\left(X_{1}^{(2)};Y_{3}^{(2)}\right) = t_2I_{\hat p}\left(X_{1}^{(2)};Y_{3}^{(2)}\right).\label{eq:Inq1}
\end{IEEEeqnarray}
where $(a)$ follows from the fact that  $U_2^{(1)}\longrightarrow X_2^{(1)} \longrightarrow Y_3^{(1)}$ forms a Markov chain and $(b)$ follows from the fact that $U_1^{(2)}\longrightarrow X_1^{(2)} \longrightarrow Y_3^{(2)}$ forms a Markov chain. Moreover as in distribution $\hat p$, $u_2^{(1)}$ and $u_1^{(2)}$ are independent from $(x_0^{(1)},x_2^{(1)})$ and $(x_0^{(2)},x_1^{(2)})$, it can be easily verified that the right-hand sides of \eqref{eq:khar}-\eqref{eq:Inq1} are equal to the right-hand sides of \eqref{eq:Back9}-\eqref{eq:Back12} with the input distribution $\hat p$, respectively. Hence, by utilizing $\hat p$ instead of $p$, the region that satisfies \eqref{eq:Back9}-\eqref{eq:Back12} is enlarged. Now, let us consider the right-hand sides of (\ref{eq:relay1})-(\ref{eq:relay22}).
\begin{IEEEeqnarray}{lll}
t_{1}I_p\left(X_{0}^{(1)};Y_{1}^{(1)}\mid X_{2}^{(1)},U_{2}^{(1)}\right) & \stackrel{(a)}{\leq}  t_{1}I_p\left(X_{0}^{(1)};Y_{1}^{(1)}\mid X_{2}^{(1)}\right) & = t_{1}I_{\hat p}\left(X_{0}^{(1)};Y_{1}^{(1)}\mid X_{2}^{(1)}\right) \label{eq:gav}\\
t_1 I_p \left(X_0^{(1)},X_2^{(1)};Y_1^{(1)}\mid U_2^{(1)}\right) & \stackrel{(a)}{\leq}  t_1 I_{p}\left(X_0^{(1)},X_2^{(1)};Y_1^{(1)}\right) & = t_1 I_{\hat p}\left(X_0^{(1)},X_2^{(1)};Y_1^{(1)}\right) \\
t_2 I_p\left(X_{0}^{(2)};Y_{2}^{(2)}\mid X_{1}^{(2)},U_{1}^{(2)}\right) & \stackrel{(b)}{\leq}  t_2 I_{p}\left(X_0^{(2)};Y_2^{(2)}\mid X_1^{(2)}\right) & = t_2 I_{\hat p}\left(X_0^{(2)};Y_2^{(2)}\mid X_1^{(2)}\right) \\
t_2 I_p\left(X_0^{(2)},X_1^{(2)};Y_2^{(2)}\mid U_1^{(2)}\right) & \stackrel{(b)}{\leq} t_2 I_{p}\left(X_0^{(2)},X_1^{(2)};Y_2^{(2)}\right) & = t_2 I_{\hat p}\left(X_0^{(2)},X_1^{(2)};Y_2^{(2)}\right) \label{eq:gus}
\end{IEEEeqnarray}
where $(a)$ follows from the fact that $U_2^{(1)}\longrightarrow (X_2^{(1)}, X_0^{(1)}) \longrightarrow Y_1^{(1)}$ form a Markov chain and $(b)$ follows from the fact that $U_1^{(2)}\longrightarrow (X_1^{(2)}, X_0^{(2)}) \longrightarrow Y_2^{(2)}$ form a Markov chain. Similarly, we observe that the right-hand sides of \eqref{eq:gav}-\eqref{eq:gus} represent the right-hand sides of inequalities (\ref{eq:relay1})-(\ref{eq:relay22}) with the input distribution $\hat p$. Hence, the region of $(R^{(1)}, R^{(2)})$ that satisfies (\ref{eq:relay1})-(\ref{eq:relay22}) and \eqref{eq:Back7}-\eqref{eq:Back12} is enlarged by utilizing the input distribution $\hat p$ instead of $p$. This proves the independency of input distributions with $u^{(1)}$ and $u^{(2)}$ in the optimum distribution.

\underline{\emph{Simplifying the achievable rate}}

As we can assume that the input distributions are of the form \eqref{eq:p1} and \eqref{eq:p2}, the achievable rate can be simplified as follows.
\begin{IEEEeqnarray}{rl}
C_{BME_{back}}^{low}& = R^{(1)} + R^{(2)} \leq\nonumber\\
&\max_{0\leq t_{1},t_{2},t_{1}+t_{2}=1}\min\left(
t_{1}I\left(X_{0}^{(1)},X_{2}^{(1)};Y_{1}^{(1)}\right),t_{2}I\left(X_{0}^{(2)},X_{1}^{(2)}
;Y_{2}^{(2)}\right)\right),\label{eq:constr1}
\end{IEEEeqnarray}

\begin{IEEEeqnarray}{rl}
&\text{subject to}\nonumber\\
&r_{Bin}^{(1)}\leq R^{(1)},\label{eq:constr3}\\
&r_{Bin}^{(2)}\leq R^{(2)},\label{eq:constr4}\\
&R^{(1)}\leq t_{1}I\left(X_{0}^{(1)};Y_{1}^{(1)}\mid X_{2}^{(1)}\right),\label{eq:constr5}\\
&R^{(2)}\leq t_{2}I\left(X_{0}^{(2)};Y_{2}^{(2)}\mid
X_{1}^{(2)}\right),\label{eq:constr6}\\
&R^{(2)}-r_{Bin}^{(2)}+r_{Bin}^{(1)}\leq t_1I\left(X_{2}^{(1)};Y_{3}^{(1)}\right),\label{eq:constr7}\\
&R^{(1)}-r_{Bin}^{(1)}+r_{Bin}^{(2)}\leq t_2I\left(X_{1}^{(2)};Y_{3}^{(2)}\right).\label{eq:constr2}
\end{IEEEeqnarray}
with input distributions
\begin{IEEEeqnarray}{rl}
&p(x_{0}^{(1)},x_{2}^{(1)})=p(x_{2}^{(1)})p(x_{0}^{(1)}|x_{2}^{(1)}),\nonumber\\
&p(x_{0}^{(2)},x_{1}^{(2)})=p(x_{1}^{(2)})p(x_{0}^{(2)}|x_{1}^{(2)}).\nonumber
\end{IEEEeqnarray}
Now, we show that (\ref{eq:constr1})-(\ref{eq:constr2}) is equivalent to
\begin{IEEEeqnarray}{rl}
C_{BME_{back}}^{low}& \leq \max_{0\leq t_{1},t_{2},t_{1}+t_{2}=1}\min\left(
t_{1}I\left(X_{0}^{(1)},X_{2}^{(1)};Y_{1}^{(1)}\right),t_{2}I\left(X_{0}^{(2)},X_{1}^{(2)}
;Y_{2}^{(2)}\right),\right.\nonumber\\
&\left.t_{1}I\left(X_{0}^{(1)};Y_{1}^{(1)}\mid X_{2}^{(1)}\right)+t_{2}I\left(X_{0}^{(2)};Y_{2}^{(2)}\mid
X_{1}^{(2)}\right),\right.\nonumber\\
&\left.t_1I\left(X_{2}^{(1)};Y_{3}^{(1)}\right)+
t_2I\left(X_{1}^{(2)};Y_{3}^{(2)}\right)\right).\label{eq:Ach}
\end{IEEEeqnarray}
First, it is easy to verify that (\ref{eq:constr1})-(\ref{eq:constr2}) imply \eqref{eq:Ach}. Now, in order to prove that the converse is also true, we show that for every possible rate $r$ satisfying \eqref{eq:Ach}, there exists a quad-tupple $\left(R^{(1)},R^{(2)}, r_{Bin}^{(1)}, r_{Bin}^{(2)} \right)$ such that $R^{(1)} + R^{(2)} = r$, $\left(R^{(1)},R^{(2)}, r_{Bin}^{(1)}, r_{Bin}^{(2)} \right)$ satisfies (\ref{eq:constr1})-(\ref{eq:constr2}), and moreover at least one of bin rates is equal to zero, i.e. $r_{Bin}^{(1)}=0$ or $r_{Bin}^{(2)}=0$.

Let us define $R^{(1)} \triangleq \min \left(r, t_{1}I\left(X_{0}^{(1)};Y_{1}^{(1)}\mid X_{2}^{(1)}\right)\right)$, $R^{(2)} \triangleq r - R^{(1)}$. As $r$ satisfies \eqref{eq:Ach}, we conclude that $(R^{(1)},R^{(2)})$ satisfies \eqref{eq:constr1}, \eqref{eq:constr5}, and \eqref{eq:constr6}. Furthermore, as $R^{(1)}+R^{(2)}=r \leq t_1I\left(X_{2}^{(1)};Y_{3}^{(1)}\right)+
t_2I\left(X_{1}^{(2)};Y_{3}^{(2)}\right)$, we conclude that either $R^{(1)} \leq t_2I\left(X_{1}^{(2)};Y_{3}^{(2)}\right)$ or $R^{(2)} \leq t_1I\left(X_{2}^{(1)};Y_{3}^{(1)}\right)$. For the sake of symmetry, let us assume that the first case has occurred, i.e. $R^{(1)} \leq t_2I\left(X_{1}^{(2)};Y_{3}^{(2)}\right)$. Now, we define $r_{Bin}^{(1)} \triangleq 0$ and $r_{Bin}^{(2)} \triangleq \max \left( 0, R^{(2)} - t_1I\left(X_{2}^{(1)};Y_{3}^{(1)}\right) \right)$. Obviously, \eqref{eq:constr3}, \eqref{eq:constr4}, and \eqref{eq:constr7} are valid. Considering \eqref{eq:constr2}, we have
\begin{equation}
R^{(1)}-r_{Bin}^{(1)}+r_{Bin}^{(2)} = R^{(1)} + \max \left( 0, r - R^{(1)} - t_1I\left(X_{2}^{(1)};Y_{3}^{(1)}\right) \right) \stackrel{(a)}{\leq} t_2I\left(X_{1}^{(2)};Y_{3}^{(2)}\right)
\end{equation}
where $(a)$ follows from the facts that $r \leq t_1I\left(X_{2}^{(1)};Y_{3}^{(1)}\right)+ t_2I\left(X_{1}^{(2)};Y_{3}^{(2)}\right)$ and $R^{(1)} \leq t_2I\left(X_{1}^{(2)};Y_{3}^{(2)}\right)$. Hence, \eqref{eq:constr2} is also valid. The second case in which $R^{(2)} \leq t_1I\left(X_{2}^{(1)};Y_{3}^{(1)}\right)$ can be dealt with in a similar manner.

Hence, from the above argument, the achievable rate of BME scheme with backward decoding can be simplified as follows:
\begin{IEEEeqnarray}{rl}
C_{BME_{back}}^{low}& \leq \max_{0\leq t_{1},t_{2},t_{1}+t_{2}=1}\min\left(
t_{1}I\left(X_{0}^{(1)},X_{2}^{(1)};Y_{1}^{(1)}\right),t_{2}I\left(X_{0}^{(2)},X_{1}^{(2)}
;Y_{2}^{(2)}\right),\right.\nonumber\\
&\left.t_{1}I\left(X_{0}^{(1)};Y_{1}^{(1)}\mid X_{2}^{(1)}\right)+t_{2}I\left(X_{0}^{(2)};Y_{2}^{(2)}\mid
X_{1}^{(2)}\right),\right.\nonumber\\
&\left.t_1I\left(X_{2}^{(1)};Y_{3}^{(1)}\right)+
t_2I\left(X_{1}^{(2)};Y_{3}^{(2)}\right)\right),
\end{IEEEeqnarray}
with probabilities
\begin{IEEEeqnarray}{rl}
&p(x_{0}^{(1)},x_{2}^{(1)})=p(x_{2}^{(1)})p(x_{0}^{(1)}|x_{2}^{(1)}),\nonumber\\
&p(x_{0}^{(2)},x_{1}^{(2)})=p(x_{1}^{(2)})p(x_{0}^{(2)}|x_{1}^{(2)}).\nonumber
\end{IEEEeqnarray}


\end{document}